\documentclass[review]{elsarticle}

\usepackage{lineno,hyperref}

\usepackage{graphicx}
\usepackage{bm}
\usepackage{multirow}
\usepackage{amsmath}
\usepackage{amsfonts}
\usepackage{amssymb}
\usepackage{amsthm}
\usepackage{float}
\usepackage{secdot}
\usepackage{tabularx}

\journal{Computers \& Fluids}

\theoremstyle{plain}
\newtheorem{theorem}{Theorem}[section]
\newtheorem{lemma}[theorem]{Lemma}
\newtheorem{corollary}[theorem]{Corollary}

\theoremstyle{definition}
\newtheorem{definition}[theorem]{Definition}

\theoremstyle{remark}
\newtheorem{remark}{Remark}

\numberwithin{equation}{section}
\numberwithin{theorem}{section}
\numberwithin{remark}{section}

\newcommand{\gbold}{\bm{g}}

\newcommand{\Tbold}{\bm{T}}
\newcommand{\nbold}{\bm{n}}

\newcommand{\Ubold}{\bm{U}}

\newcommand{\ubold}{\bm{u}}

\newcommand{\vbold}{\bm{v}}
\newcommand{\wbold}{\bm{w}}

\newcommand{\xbold}{\bm{x}}

\newcommand{\ipt}[2]{\left(#1,#2\right)_{\mathcal{T}_h}}

\newcommand{\ipbt}[2]{\left\langle#1,#2\right\rangle_{\partial \mathcal{T}_h}}

\newcommand{\ipbf}[2]{\left\langle#1,#2\right\rangle_{\mathcal{F}_h}}
\newcommand{\iipbf}[2]{\left\langle#1,#2\right\rangle_{\mathcal{F}_h^i}}

\newcommand{\llbracket}{\left[\!\left[}
\newcommand{\rrbracket}{\right] \! \right]}

\newcommand{\llcurve}{\left\{\!\left\{}
\newcommand{\rrcurve}{\right\} \! \right\}}

\begin{document}

\begin{frontmatter}

\title{Versatile Mixed Methods for Non-Isothermal Incompressible Flows}

\author{Edward A. Miller}
\author{Xi Chen}
\author{David M. Williams \corref{mycorrespondingauthor}} 
\address{Department of Mechanical Engineering, The Pennsylvania State University, University Park, Pennsylvania 16802}

\cortext[mycorrespondingauthor]{Corresponding author}
\ead{david.m.williams@psu.edu}

\begin{abstract}
The purpose of this paper is to extend the versatile mixed methods originally developed by Chen and Williams for isothermal flows in ``Versatile Mixed Methods for the Incompressible Navier-Stokes Equations," Computers \& Mathematics with Applications, 2020, (under review), to simulate non-isothermal incompressible flows. These new mixed methods are particularly interesting, as with only minor modifications they can be applied to a much broader range of flows, including non-isothermal weakly-compressible flows, and fully-compressible flows. In the main body of this paper, we carefully develop these mixed methods for solving the Boussinesq model equations. Thereafter, we prove the L2-stability of the discrete temperature field, and assess the practical behavior of the methods by applying them to a set of well-known convection problems. 
\end{abstract}


\end{frontmatter}

\section{Introduction} \label{sec;introduction}

The motion of a non-isothermal incompressible fluid is frequently induced by buoyancy forces, viscous forces, and pressure fields. In accordance with standard practices, we refer to the motion that is induced solely by buoyancy forces as \emph{natural} or \emph{free} convection, the motion that is induced solely by viscous forces and pressure fields as \emph{forced} convection, and the motion that is induced by all three factors as \emph{mixed} convection. In order to characterize the various types of convection, one may solve the incompressible Navier-Stokes equations for mass and momentum conservation, in conjunction with a temperature equation (usually obtained from the internal energy or enthalpy equations). In addition, one may couple the momentum and temperature equations via the approach of Oberbeck~\cite{oberbeck1879warmeleitung} and Boussinesq~\cite{boussinesq1903theorie} by adding a temperature-dependent buoyancy term to the RHS of the momentum equation. The buoyancy term is assumed to be directly proportional to changes in the temperature field, and these changes are assumed to be small enough such that the density remains constant. This approximation is frequently referred to as the Boussinesq model~\cite{zeytounian2003joseph}, or (less commonly) the Oberbeck-Boussinesq model~\cite{dallmann2016stabilized}. For practical applications, it is usually necessary to solve the Boussinesq model in the vicinity of complicated geometries, using unstructured meshes. As a result, our preference is to use finite element methods for solving the model because of their ability to operate on both structured and unstructured meshes, while simultaneously achieving high-order accuracy, stability, and robustness.  

In what follows, we briefly review some previous efforts to apply finite element methods to the Boussinesq model. Some of the earliest work in this area was performed by Laskaris~\cite{laskaris1975finite} who used a high-order continuous Galerkin (CG) method to simulate channel flows with heated walls. In addition, Young et al.~\cite{young1976steady,young1976unsteady} and Tabarrok and Lin~\cite{tabarrok1977finite} used a similar approach to study natural convection in heated cavities. Next, Gartling~\cite{gartling1977convective} used a CG method to simulate a thin-walled tube with wall heat transfer, a rectangular heat exchanger, and a heated hexagonal cylinder in a cooled cavity. Thereafter, Marshall et al.~\cite{marshall1978natural} used a high-order CG method with a penalty function (for enforcement of the dilatational constraint) to simulate a heated cavity. This was the first time that a finite element method was successfully applied to natural convection problems for a wide range of Rayleigh numbers ($10^4 -10^7$). Based on this work, Reddy and Satake~\cite{reddy1980comparison} formulated an alternative CG method, and used it to simulate heated, non-convex, straight-sided cavities. It is important to note that all of the early work described above was limited to two-dimensional geometries. Fortunately, with the advent of more powerful computers and more advanced stabilization strategies, such as the Galerkin Least Squares (GLS) approach~\cite{hughes1987new,hughes1989new,baiocchi1993virtual}, the solutions to three-dimensional problems became possible. Some of the early work in this area was performed by Tang and Tsang~\cite{tang1994least,tang1997temporal}, who used least-square finite element methods to simulate  three-dimensional heated cavities, and accurately reproduce the dynamics of Rayleigh-B\'enard convection cells. A detailed review of the latest efforts to apply finite element methods to natural and mixed convection problems is beyond the scope of the present article. However, the interested reader may consult~\cite{reddy2010finite,dallmann2015finite} for an extensive collection of references on this topic. 

Despite the many applications of finite element methods to the Boussinesq model, there have been a relatively small number of efforts to rigorously analyze the existing methods, or to develop new mixed methods which maintain inf-sup stability. Some pioneering efforts in this area were undertaken by Boland and Layton~\cite{boland1990error,boland1990analysis}, as they derived stability and error estimates for CG methods for steady and unsteady natural convection problems. In addition, they analyzed low-order, non-conforming discontinuous Galerkin (DG) methods. Most notably, they were among the first researchers to recognize the importance of using a skew-symmetrizing procedure to stabilize the convective operator in the temperature equation. Subsequently, their work was expanded by Dorok et al.~\cite{dorok1994aspects} and Bernardi et al.~\cite{bernardi1995couplage}, who developed stability and error estimates for mixed methods. More recently, Codina et al.~\cite{codina2010finite} and L\"{o}we and Lube~\cite{lowe2012projection} developed variational multiscale (VMS) methods for problems with turbulent mixed convection. Within the VMS framework, they constructed rigorous stability estimates and (in the case of~\cite{lowe2012projection}) error estimates for the resulting schemes. Thereafter, Dallmann and Arndt~\cite{dallmann2015finite,dallmann2016stabilized} developed a mixed method which was stabilized using a combination of local projection stabilization~\cite{roos2008robust,matthies2015local}, streamline-upwind stabilization~\cite{brooks1982streamline,hughes1986new,hughes1987recent}, and grad-div stabilization~\cite{franca1988two}. For this method, they rigorously derived stability and error estimates, and produced accurate numerical results for a wide range of steady and unsteady convection problems. Next, Rebollo et al.~\cite{rebollo2018high} developed a mixed method which they stabilized using an interpolation-based operator that acts as a low-pass filter. We note that, although the performance of this method is quite adequate from an accuracy standpoint, it is only weakly consistent. Most recently, de Frutos et al.~\cite{de2019grad} derived an optimal set of stability and error estimates for grad-div stabilized, inf-sup stable mixed methods. These methods are effectively a subset of the methods constructed by Dallmann and Arndt in~\cite{dallmann2015finite,dallmann2016stabilized}. Lastly, we note that there are ongoing efforts to analyze mixed methods for Boussinesq models with temperature-dependent parameters (cf.~\cite{oyarzua2017analysis,almonacid2018mixed,allendes2018divergence,almonacid2020new} for several examples).

Due to the limited number of efforts to develop mixed methods (see above), there are still opportunities to improve the robustness, accuracy, and flexibility of the methods. With this in mind, the goal of the present paper is to extend the recently developed versatile mixed methods (see~\cite{chen2020versatile}) to solve the Boussinesq model with constant parameters. For the sake of completeness, let us briefly describe the underlying philosophy of versatile mixed methods: i) we begin with the compressible formulation of the governing equations and then enforce the assumption of constant density, ii) we maintain the presence of dilatational terms (and similar terms) that would usually be neglected, and iii) we discretize the resulting formulation using standard, inf-sup stable, mixed methods. This approach has several advantages, as most importantly, it can be immediately applied to weakly-compressible flows, and furthermore, it ensures that the dilatational constraint is enforced in a consistent fashion in each of the governing equations. In~\cite{chen2020versatile}, this philosophy was applied to the isothermal incompressible Navier-Stokes equations. There, we used the full compressible stress tensor (with the dilatational component) in the momentum conservation equation, and we rigorously proved the stability of the discrete velocity field. The resulting methods were successfully applied to isothermal Taylor-Green and Gresho vortex problems. In this work, we apply the same methods to non-isothermal incompressible flows.

The format of this paper is as follows. In section 2, we formally introduce the Boussinesq model equations for non-isothermal incompressible flows and we develop the notation and mathematical machinery for discretizing these equations. In sections 3 and 4, we introduce the versatile mixed methods, and prove the stability of the discrete temperature field. In section 5, we apply the methods to a set of standard benchmark problems involving natural and mixed convection. Finally, in section 6, we conclude with a summary of our work and a few final remarks.


\section{Preliminaries} \label{prelim_section}

Let us start by introducing a domain $\Omega_t = \left(0, t_n \right) \times \Omega$, where $\Omega \in \mathbb{R}^{d}$ is a spatial domain and $\left(0, t_n\right) \in \mathbb{R}$ is a temporal domain. In a natural fashion, we denote the spatial and temporal coordinates by $\xbold$ and $t$, and we denote the spatial and temporal derivatives by $\nabla \left(\cdot\right)$ and $\partial_t \left(\cdot \right)$, respectively. We assume $d = 2$ or 3, and that the domain boundary $\partial \Omega$ is composed from straight line segments (for the case of $d =2$) and planar faces (for the case of $d=3$). Inside the domain $\Omega_t$, we are interested in simulating the motion of a homogeneous, non-isothermal, incompressible fluid with a constant density $\rho_0$, and non-constant velocity, temperature, and pressure fields $\ubold = \ubold \left(t, \xbold\right)$, $T = T \left(t, \xbold\right)$, and $p = p \left(t, \xbold\right)$. Since the density is constant, we find it convenient to divide the governing equations by $\rho_0$, and then introduce density-weighted quantities, such as $\widetilde{p} = p/\rho_0$ (the kinematic pressure). We introduce the tilde symbol to avoid abuses of notation which can result from ignoring differences between density scaled and unscaled quantities. Now, having established the necessary background, we present the Boussinesq model for non-isothermal flows
\begin{align}
&\nabla \cdot \ubold = 0, &\text{in} \; \, \Omega_t \label{mass_cons} \\[1.5ex]
&\partial_t \, \ubold + \nabla \cdot \left( \ubold \otimes \ubold + \widetilde{p} \, \mathbb{I} \right) - \nabla \cdot \widetilde{\bm{\tau}} = - \beta T \bm{g} + \widetilde{\bm{f}}_{\ubold},  & \text{in} \; \, \Omega_t \label{moment_cons}\\[1.5ex]
&\partial_t T + \nabla \cdot \left( T  \ubold  \right) - \nabla \cdot \left( \alpha \gamma \nabla T  \right) = \frac{1}{C_v} \Big[ \widetilde{\bm{\tau}} : \nabla \ubold -\widetilde{p} \left(\nabla \cdot \ubold\right) \Big] + \widetilde{f}_{T}, &\text{in} \; \, \Omega_t. \label{energy_cons} 
\end{align}
These equations are subject to the following boundary and initial conditions 
\begin{align}
&\ubold=0, &\text{on} \; \, \partial \Omega_t,\\[1.5ex]
&T=0, &\text{on} \; \, \partial \Omega_t,\\[1.5ex]
&\ubold(0,\xbold)=\ubold_0(\xbold), &\text{in} \; \, \Omega, \\[1.5ex]
&T(0,\xbold)=T_0(\xbold), &\text{in} \; \, \Omega.
\end{align}
Furthermore, in order to close the equations, we define $\widetilde{\bm{\tau}}$ as the stress tensor
\begin{align}
\widetilde{\bm{\tau}} = \nu \left( \nabla \ubold + \nabla \ubold^T - \frac{2}{3} \left(\nabla \cdot \ubold \right) \mathbb{I} \right), \label{stress_tensor}
\end{align}
$\bm{g}$ as the gravitational acceleration (where $g_i = -g \delta_{id}$ with $g = const$), $\widetilde{\bm{f}}_{\ubold}$ as a source term for the linear momentum,  $\widetilde{f}_{T}$ as a source term for the temperature, $C_v$ as the specific heat at constant volume, $C_p$ as the specific heat at constant pressure, $\gamma = C_p/C_v$ as the ratio of specific heats, $\alpha = \kappa/\left(C_p \, \rho_0\right)$ as the thermal diffusivity coefficient,  $\beta$ as the thermal expansion coefficient, $\kappa$ as the thermal conductivity coefficient, $\nu = \mu/\rho_0$ as the kinematic viscosity coefficient, and $\mu$ as the dynamic viscosity coefficient.

Before proceeding further, it is important to note that our equations for the temperature and the stress tensor (Eqs.~\eqref{energy_cons} and~\eqref{stress_tensor}) are unconventional. In particular, it is common practice to neglect the viscous dissipation and pressure work terms on the RHS of Eq.~\eqref{energy_cons}, such that
\begin{align}
     \partial_t T + \nabla \cdot \left( T  \ubold  \right) - \nabla \cdot \left( \alpha \gamma \nabla T  \right) = \widetilde{f}_{T}. \label{energy_cons_standard}
\end{align}
In addition, most researchers neglect the divergence and gradient transpose terms on the RHS of Eq.~\eqref{stress_tensor}, as follows
\begin{align}
    \widetilde{\bm{\tau}} = \nu \nabla \ubold. \label{stress_tensor_standard}
\end{align}
However, we prefer to use Eqs.~\eqref{energy_cons} and \eqref{stress_tensor} due to their superior physical accuracy, flexibility, and discrete consistency. We refer the interested reader to~\cite{chen2020versatile} for a detailed discussion of our motivation for using the full stress tensor (Eq.~\eqref{stress_tensor}). In what follows, we will only discuss our motivation for using the full temperature equation (Eq.~\eqref{energy_cons}). 

\begin{enumerate}
    \item The formulation in Eq.~\eqref{energy_cons} contains the viscous dissipation term, and thereby successfully captures the physical conversion of kinetic energy into internal energy (heat). Of course, the viscous dissipation term will be small in most incompressible flows, however it will rarely completely vanish. Therefore, by neglecting this term, we introduce a small but unnecessary amount of error into the final solution. Furthermore, this error is difficult to control, as it does not vanish in the asymptotic limit as the element size goes to zero, or the polynomial order goes to infinity.
    \item The formulation in Eq.~\eqref{energy_cons} is more suitable for adaptation to \emph{compressible} flows, as it retains the pressure work and viscous dissipation terms which become increasingly important in these types of flows. Retaining these terms helps facilitate flexibility of the resulting methods, and encourages code-reuse between incompressible and compressible CFD~codes. 
    \item The formulation in Eq.~\eqref{energy_cons} is more `consistent', as it enables consistent enforcement of the dilatational constraint. In order to see this, we begin by noting that Eq.~\eqref{energy_cons} retains the pressure work term, which is guaranteed to vanish at the continuous level (by Eq.~\eqref{mass_cons}), but which may or may not vanish at the discrete level. Evidently, for pointwise divergence-free methods, the pressure term vanishes in both cases, but for more general methods, the dilatation term typically only vanishes in the weak sense, and the pressure term is non-zero. Therefore, neglecting the pressure term \emph{a priori} is inconsistent, as this effectively forces the dilatation contribution to vanish pointwise in the temperature equation, even though it may only vanish weakly in the mass conservation equation. Naturally, we prefer to use Eq.~\eqref{energy_cons}, as it avoids this inconsistency. 
\end{enumerate}
In summary, we have introduced a `versatile' approach in which we solve Eqs.~\eqref{mass_cons}--\eqref{energy_cons} in conjunction with the stress tensor in Eq.~\eqref{stress_tensor}. In what follows, we will introduce the necessary machinery for discretizing these equations.

In accordance with standard practices, we tessellate the spatial domain $\Omega$ with a mesh $\mathcal{T}_h$. The mesh is composed from straight-sided, $d$-dimensional triangular or cubic elements $K$, with characteristic size $h$. The faces of elements on the perimeter of the mesh are required to exactly conform to the domain boundaries, and the union of all the elements is required to cover the domain. In addition, for the sake of simplicity the elements are required to be non-overlapping, and the mesh is required to be devoid of hanging nodes. The boundary of each element $K$ is denoted by $\partial K$ and the outward-pointing unit normal vector on this boundary is denoted by $\nbold$. Elements are considered to be `face neighbors' if they share a $(d-1)$-dimensional face $F$. We denote the unit normal vector that points from the positive side to the negative side of the shared face as $\nbold_{+}$, and naturally $\nbold_{-} = -\nbold_{+}$. The total collection of faces in the mesh is denoted by $\mathcal{F}_h$, and the faces of a single element $K$ are denoted by $\mathcal{F}_{K} = \left\{ F \in \mathcal{F}_h : F \subset \partial K \right\}$. The set of interior faces is denoted by $\mathcal{F}_h^i = \{F \in \mathcal{F}_h : F \cap \partial\Omega = \emptyset \}$ and the set of boundary faces by $\mathcal{F}_h^{\partial} = \{F \in \mathcal{F}_h : F \cap \partial\Omega \neq \emptyset \}$. Finally, for a given face $F$, we can define a normal vector $\nbold_F$ which points from the positive to the negative side of the face. 

Next, one may define jump $\llbracket \cdot \rrbracket$ and average $\llcurve \cdot \rrcurve$ operators for an interior face $F \in \mathcal{F}_h^i$ as follows
\begin{align*}
\llbracket \phi \rrbracket &= \phi_{+} - \phi_{-}, \qquad \llbracket \phi \nbold \rrbracket = \phi_{+} \nbold_{+} + \phi_{-} \nbold_{-}, \qquad \llcurve \phi \rrcurve = \frac{1}{2} \left( \phi_{+} + \phi_{-} \right), \\[1.5ex]
\llbracket \vbold \rrbracket &= \vbold_{+} - \vbold_{-}, \qquad \llbracket \vbold \otimes \nbold \rrbracket = \vbold_{+} \otimes \nbold_{+} + \vbold_{-} \otimes \nbold_{-}, \qquad  \llcurve \vbold \rrcurve = \frac{1}{2} \left( \vbold_{+} + \vbold_{-} \right),
\end{align*}
where $\phi$ is a generic scalar function, and $\vbold$ is a generic vector function. Similarly, for all boundary faces $F\in \mathcal{F}_h^{\partial}$, one may define
\begin{align*}
\llbracket \phi \rrbracket &= \phi, \qquad \llbracket \phi \nbold \rrbracket = \phi \nbold, \qquad \llcurve \phi \rrcurve = \phi, \\[1.5ex]
\llbracket \vbold \rrbracket &= \vbold, \qquad \llbracket \vbold \otimes \nbold \rrbracket = \vbold \otimes \nbold, \qquad  \llcurve \vbold \rrcurve = \vbold.
\end{align*}

In addition, it is convenient to introduce some standard notation for representing inner products. With this in mind, let us introduce a generic vector $\wbold$ and generic tensors $\Tbold$ and $\Ubold$. Note: here, we assume that $\vbold$, $\wbold$, $\Tbold$, $\Ubold$, and $\phi$ are sufficiently smooth, such that the associated integrations are possible. Based on this assumption, we can define
\begin{align*}
\ipt{\vbold}{\wbold} & = \sum_{K \in \mathcal{T}_h} \int_{K} \vbold \cdot \wbold \, dV, \qquad \ipt{\Tbold}{\Ubold} = \sum_{K \in \mathcal{T}_h} \int_{K} \Tbold : \Ubold \, dV, \\[1.5ex]
\ipbt{\vbold}{\wbold} &= \sum_{K \in \mathcal{T}_h} \int_{\partial K} \vbold \cdot \wbold \, dA, \qquad \ipbt{\Tbold}{\Ubold} = \sum_{K \in \mathcal{T}_h} \int_{\partial K} \Tbold : \Ubold \, dA, \\[1.5ex]
\ipbf{\vbold}{\wbold} & = \sum_{F \in \mathcal{F}_h} \int_{F} \vbold \cdot \wbold \, dA, \qquad \ipbf{\Tbold}{\Ubold} = \sum_{F \in \mathcal{F}_h} \int_{F} \Tbold : \Ubold \, dA.
\end{align*}
Using this notation, we can introduce the well-known integration by parts formulas
\begin{align*}
\int_{\partial K} \phi \left( \vbold \cdot \nbold \right) dA & = \int_{K} \left(\phi \left(\nabla\cdot \vbold \right) + \vbold \cdot \nabla \phi  \right) dV, \\[1.5ex]
\int_{\partial K} \vbold \cdot \Tbold \nbold \, dA &= \int_{K} \left(\vbold \cdot \left( \nabla \cdot \Tbold \right) + \Tbold : \nabla \vbold \right) dV,
\end{align*}
which can be rewritten as
\begin{align*}
\left\langle \phi \vbold, \nbold \right\rangle_{\partial K} &= \left( \phi, \nabla \cdot \vbold \right)_{K} + \left( \vbold, \nabla \phi \right)_{K}, \\[1.5ex]
\left\langle \vbold, \bm{T} \nbold \right\rangle_{\partial K} &= \left(\vbold, \nabla \cdot \bm{T} \right)_{K} + \left(\bm{T}, \nabla \vbold \right)_{K}.
\end{align*}
In what follows, we will conclude this section by defining the standard function spaces for mixed finite element methods. We start by introducing the broken Sobolev space
\begin{align*}
    & \bm{W}^{m,p} (\mathcal{T}_h) = \left\{\wbold  \in \bm{L}^{p}(\Omega), \wbold|_{K} \in\bm{W}^{m,p}(K),~\forall K \in \mathcal{T}_h \right\},
\end{align*}
where $\bm{W}^{m,p} \left(\mathcal{T}_h \right) = \left(W^{m,p} \left(\mathcal{T}_h\right) \right)^d$. Next, we introduce the Hilbert spaces
\begin{align*}
   &\bm{H}_{0}(\text{div};\Omega)= \left\{ \wbold  :  \wbold\in \bm{L}^{2}(\Omega),~\nabla\cdot\wbold \in L^{2}(\Omega),~\wbold \cdot\nbold|_{\partial\Omega}=0\right\},\\[1.5ex] 
   &\bm{H}_{0}^{1}(\Omega)= \left\{ \wbold  :  \wbold\in \bm{H}^{1}(\Omega),~\wbold|_{\partial\Omega}=0\right\},
\end{align*}
where $\bm{H}^{1} \left(\Omega\right) = \left(H^{1} \left(\Omega\right) \right)^d$. Having established these spaces, we can define scalar-valued polynomial spaces $Q_{h}^{DC}$ and $Q_{h}^{C}$ for the pressure, and $R_{h}^{C}$ for the temperature
\begin{align*}
    &Q_h^{DC} = \left\{ q_h  : q_h \in L^2_{\ast} \left( \Omega \right), q_{h} |_{K} \in P_{k} \left( K \right), \forall K \in \mathcal{T}_h \right\}, \\[1.5ex]
    &Q_h^{C} = \left\{ q_h  : q_h \in C^{0} \left( \Omega \right), q_{h} |_{K} \in P_{k} \left( K \right), \forall K \in \mathcal{T}_h \right\} \cap L_{\ast}^2 \left(\Omega \right),
    \\[1.5ex]
    &R_h^{C} = \left\{ r_h  : r_h \in C^0 \left( \Omega \right), r_h |_{K} \in P_{k} \left( K \right), \forall K \in \mathcal{T}_h \right\} \cap H_{0}^1 \left(\Omega \right),
\end{align*}
where $P_{k} \left(K \right)$ is the space of polynomials of degree $\leq k$, and $L^2_{\ast} \left(\Omega\right)$ is the space of $L^2$ functions with zero mean. Furthermore, we can define the vector-valued Raviart-Thomas and Taylor-Hood spaces for the velocity
\begin{align*}
    &\bm{W}_h^{RT}  = \left\{ \wbold_h : \wbold_h \in \bm{H}_{0}\left(\text{div}; \Omega \right), \wbold_h |_{K} \in  \bm{RT}_k \left( K \right), \forall K \in \mathcal{T}_h \right\}, \\[1.5ex]
    &\bm{W}_h^{TH} = \left\{ \wbold_h : \wbold_h \in \bm{C}^{0} \left(\Omega\right), \wbold_h |_{K} \in  \left(P_{k+1} \left( K \right) \right)^d, \forall K \in \mathcal{T}_h \right\} \cap \bm{H}_{0}^{1}(\Omega),
\end{align*}
where $\bm{C}^{0} \left(\Omega\right) = \left(C^{0} \left(\Omega\right) \right)^d$, and
\begin{align*}
\bm{RT}_k \left(K \right) = \left(P_k \left(K \right) \right)^d \oplus P_k \left(K \right) \xbold.
\end{align*}
Lastly, we can introduce $\bm{W}_h^{BDM}$, the Brezzi-Douglas-Marini space (see~\cite{boffi2013mixed} for an explicit definition of this space).

\section{Versatile Mixed Methods}

In this section, we develop a general class of mixed methods for solving  Eqs.~\eqref{mass_cons} -- \eqref{energy_cons}. The methods can be constructed using the following steps: i) choose function spaces $Q_h \subset L^2_{\ast} \left(\Omega \right)$, $R_h \subset H^1_0 \left(\Omega\right)$, and $\bm{W}_h \subset \bm{H}_{0}(\text{div};\Omega)$, ii) identify test functions $\left(q_h, r_h, \wbold_h \right)$ that span $Q_h \times R_h \times \bm{W}_h$, and iii) find unknowns $\left(\widetilde{p}_h, T_h, \ubold_h \right)$ in $Q_h  \times R_h \times \bm{W}_h$ that satisfy 
\begin{align}
& \ipt{ \nabla \cdot \ubold_h}{q_h} = 0, \label{mass_cons_disc}
\end{align}
\begin{align}
\nonumber & \ipt{\partial_t \ubold_h }{\wbold_h} - \ipt{ \ubold_h \otimes \ubold_h}{\nabla_h \wbold_h} - \ipt{\widetilde{p}_h}{\nabla \cdot \wbold_h} + \ipbt{\widehat{\bm{\sigma}}_{\text{inv}} \, \nbold}{\wbold_h} \\[1.5ex]  
\nonumber & + \nu_h \bigg[ \ipt{\nabla_h \ubold_h + \nabla_h \ubold_{h}^{T} - \frac{2}{3} \left(\nabla \cdot \ubold_h \right) \mathbb{I}}{\nabla_h \wbold_h} - \ipbt{\widehat{\bm{\sigma}}_{\text{vis}} \, \nbold}{\wbold_h} \\[1.5ex]
\nonumber & + \ipbt{\widehat{\bm{\varphi}}_{\text{vis}} - \ubold_h}{\left( \nabla_h \wbold_h + \nabla_h \wbold_{h}^{T} - \frac{2}{3} \left(\nabla \cdot \wbold_h \right) \mathbb{I} \right) \nbold} \bigg] -\frac{1}{2} \ipt{\left(\nabla \cdot \bm{u}_h \right) \ubold_h}{\wbold_h} \\[1.5ex]
&=  - \ipt{\beta_h T_h \gbold}{\wbold_h} +  \ipt{\widetilde{\bm{f}}_{\ubold}}{\wbold_h}, \label{moment_cons_disc}
\end{align}
\begin{align}
& \nonumber \ipt{\partial_t T_h}{r_h} - \ipt{T_h \ubold_h}{\nabla_h r_h} + \ipbt{\widehat{\bm{\phi}}_{\text{inv}}  \cdot \nbold}{r_h}  \\[1.5ex]
\nonumber & + \alpha_h \gamma_h \Bigg[ \ipt{\nabla_h T_h}{\nabla_h r_h}  -\ipbt{\widehat{\bm{\phi}}_{\text{vis}}  \cdot \nbold}{r_h} \\[1.5ex]
\nonumber &+ \ipbt{\widehat{\lambda}_{\text{vis}} - T_h}{\nabla_h r_h \cdot \nbold} \Bigg] -\frac{1}{2} \ipt{\left(\nabla \cdot \ubold_h \right) T_h}{r_h} \\[1.5ex]
\nonumber &  = \ipt{\frac{1}{C_v} \Bigg[ \nu_h \left(\nabla_h \ubold_h + \nabla_h \ubold_{h}^{T} - \frac{2}{3} \left(\nabla \cdot \ubold_h \right) \mathbb{I} \right):\nabla_h \ubold_h - \widetilde{p}_h \left(\nabla \cdot \ubold_h \right)\Bigg]}{r_h} \\[1.5ex]
& + \ipt{\widetilde{f}_{T}}{r_h}. \label{energy_cons_disc}
\end{align}
Here, we note that the quantities with hats (for example~$\widehat{\bm{\sigma}}_{\text{inv}}$) denote numerical fluxes. An array of possible formulas for the fluxes are given in section~\ref{numerical_flux_section} of the Appendix. By substituting these formulas into Eqs.~\eqref{mass_cons_disc} -- \eqref{energy_cons_disc}, one may rewrite the equations in standard form as follows
\begin{align}
& b_h \left(\ubold_h, q_h \right) = 0, \label{incomp_form_one_A} \\[1.5ex]
\nonumber & \ipt{\partial_t \, \ubold_h}{\wbold_h} + c_h \left(\ubold_h; \ubold_h, \wbold_h \right) +\nu_h a_h \left(\ubold_h, \wbold_h \right)  - b_h \left( \wbold_h, \widetilde{p}_h \right) \\[1.5ex]
&= - \ipt{\Xi \left(T_h\right)}{\wbold_h} + \ipt{\widetilde{\bm{f}}_{\ubold}}{\wbold_h}, \label{incomp_form_one_B}  \\[1.5ex]
&\nonumber \ipt{\partial_t T_h}{r_h} + \underline{c}_h\left(\ubold_h; T_h, r_h\right) + \alpha_h \gamma_h \, \underline{a}_h \left(T_h, r_h \right) \\[1.5ex]
&= \ipt{\Phi \left(\ubold_h\right) + \Psi \left(\ubold_h, \widetilde{p}_h \right)}{r_h} +  \ipt{\widetilde{f}_{T}}{r_h}. \label{incomp_form_one_C} 
\end{align}
Next, we must define the operators $a_h$, $b_h$, $c_h$, $\underline{a}_h$, $\underline{c}_h$, $\Xi$, $\Phi$, and $\Psi$. In order to setup these definitions, we introduce functions $q_h \in Q_h$, $r_h$ and $\theta_h \in R_h$, and $\vbold_h, \wbold_h$ and  $\bm{\xi}_h \in \bm{W}_h$. Thereafter, we expand the operators in Eqs.~\eqref{incomp_form_one_A} and~\eqref{incomp_form_one_B} as follows
\begin{align}
    b_h \left(\vbold_h, q_h \right) &= \ipt{\nabla \cdot \vbold_h}{q_h}, \label{bilinear_press_div} \\[1.5ex]
    c_h \left(\bm{\xi}_h; \vbold_h, \wbold_h \right) &= \ipt{\bm{\xi}_h \cdot \nabla_h \vbold_h}{\wbold_h} + \frac{1}{2} \ipt{\left(\nabla \cdot \bm{\xi}_h \right) \vbold_h}{\wbold_h} \label{trilinear} \\[1.5ex]
    & \nonumber - \iipbf{ \left( \bm{\xi}_h \cdot \nbold_F \right) \llbracket \vbold_h \rrbracket}{\llcurve \wbold_h \rrcurve} + \zeta \iipbf{\left| \bm{\xi}_h \cdot \nbold_F \right|\llbracket \vbold_h \rrbracket}{\llbracket \wbold_h \rrbracket}, 
\end{align}
\begin{align}
    a_h \left(\vbold_h, \wbold_h \right) &=  \ipt{ \nabla_h \vbold_h + \nabla_h \vbold_{h}^{T} - \frac{2}{3} \left(\nabla \cdot \vbold_h \right) \mathbb{I}}{\nabla_h \wbold_h} \label{diff_bilinear} \\[1.5ex]
    \nonumber & -\ipbf{\llbracket \vbold_h \rrbracket}{\llcurve  \nabla_h \wbold_h + \nabla_h \wbold_{h}^{T} - \frac{2}{3} \left(\nabla \cdot \wbold_h \right) \mathbb{I} \rrcurve \nbold_F}  \\[1.5ex]
    \nonumber & -\ipbf{\llbracket \wbold_h \rrbracket}{\llcurve  \nabla_h \vbold_h + \nabla_h \vbold_{h}^{T} - \frac{2}{3} \left(\nabla \cdot \vbold_h \right) \mathbb{I} \rrcurve \nbold_F} + \ipbf{\frac{\eta}{h_F} \llbracket \vbold_h \rrbracket}{\llbracket \wbold_h \rrbracket},
\end{align}
\begin{align}
    \Xi \left(r_h\right) = \beta_h r_h \gbold. \label{grav_term}
\end{align}
In addition, the operators in Eq.~\eqref{incomp_form_one_C} can be expanded as follows
\begin{align}
\underline{c}_h \left(\bm{\xi}_h; \theta_h, r_h \right) &= \ipt{\bm{\xi}_h \cdot \nabla_h \theta_h}{r_h} + \frac{1}{2} \ipt{\left(\nabla \cdot \bm{\xi}_h \right) \theta_h}{r_h} \label{trilinear_temp}\\[1.5ex]
& \nonumber - \iipbf{ \left( \bm{\xi}_h \cdot \nbold_F \right) \llbracket \theta_h \rrbracket}{\llcurve r_h \rrcurve} + \delta \iipbf{\left| \bm{\xi}_h \cdot \nbold_F \right|\llbracket \theta_h \rrbracket}{\llbracket r_h \rrbracket}, 
\end{align}
\begin{align}
    \underline{a}_h \left(\theta_h, r_h\right) &= \ipt{\nabla_h \theta_h}{\nabla_h r_h}  -\ipbf{\llbracket \theta_h \rrbracket}{\llcurve \nabla_h r_h \rrcurve \cdot \nbold_F}  \label{diff_bilinear_temp} \\[1.5ex]
    \nonumber &-\ipbf{\llbracket r_h \rrbracket}{\llcurve \nabla_h \theta_h \rrcurve \cdot  \nbold_F} + \ipbf{\frac{\varepsilon}{h_F} \llbracket \theta_h \rrbracket}{\llbracket r_h \rrbracket},
\end{align}
\begin{align}
    \Phi \left(\vbold_h\right) &= \frac{\nu_h}{C_v} \left(  \left(\nabla_h \vbold_h + \nabla_h \vbold_{h}^{T} - \frac{2}{3} \left(\nabla \cdot \vbold_h \right) \mathbb{I} \right):\nabla_h \vbold_h \right), \label{visc_dissipation} \\[1.5ex]
    \Psi \left(\vbold_h, q_h\right) & = - \frac{1}{C_v} \left( q_h \left(\nabla \cdot \vbold_h \right) \right).
\end{align}
It is possible to simplify these expressions in the particular case when the method is pointwise divergence-free. One may consult section~\ref{div_free_section} of the Appendix for details.

\section{Analysis of Versatile Mixed Methods}

In this section, we rigorously analyze the stability of the versatile mixed methods which were introduced in section 3. In preparation for this analysis, we first define a special set of norms on broken Sobolev spaces. Thereafter, we establish the coercivity of the bilinear form $\underline{a}_h$ (Eq.~\eqref{diff_bilinear_temp}) and the semi-coercivity of the trilinear form $\underline{c}_h$ (Eq.~\eqref{trilinear_temp}). Next, we use these results to prove the L2-stability of the discrete temperature field. Finally, we discuss the relationship between the stability properties of the discrete temperature and velocity fields.

\subsection{Norm Definitions}

\begin{definition}[Gradient Norm] \label{norm_def_standard}
    Consider the scalar-valued function $r \in W^{1,p} \left(\mathcal{T}_h\right)$. Then,
    \begin{align*}
        \left\| r \right\|_{\text{grad},p} &= \left[ \left\| \nabla_h r \right\|_{\bm{L}^p\left(\Omega\right)}^{p} + \sum_{F \in \mathcal{F}_h} \frac{1}{h_{F}^{p-1}} \left\| \llbracket r \rrbracket \right\|_{L^{p} \left(F\right)}^{p} \right]^{1/p} \\[1.5ex]
        & = \left[ \sum_{K \in \mathcal{T}_h} \int_{K} \left( \sum_{j}^d \left| \partial_{j} r \right|^p \right) dV + \sum_{F \in \mathcal{F}_h} \frac{1}{h_{F}^{p-1}} \int_{F}  \left| \llbracket r \rrbracket \right|^{p} dA \right]^{1/p},
    \end{align*}
    is a norm on $\Omega$. In a similar fashion, for the vector-valued function $\wbold \in \bm{W}^{1,p} \left(\mathcal{T}_h\right)$, we have
    \begin{align}
        \nonumber \left\| \wbold \right\|_{\text{grad},p} &= \left[ \left\| \nabla_h \wbold \right\|_{\bm{L}^p\left(\Omega\right) \times \bm{L}^p\left(\Omega\right)}^{p} + \sum_{F \in \mathcal{F}_h} \frac{1}{h_{F}^{p-1}} \left\| \llbracket \wbold \rrbracket \right\|_{\bm{L}^{p} \left(F\right)}^{p} \right]^{1/p} \\[1.5ex]
        & = \left[ \sum_{K \in \mathcal{T}_h} \int_{K} \left( \sum_{i,j}^d \left| \partial_{j} w_i \right|^p \right) dV + \sum_{F \in \mathcal{F}_h} \frac{1}{h_{F}^{p-1}} \int_{F} \left( \sum_{i}^{d} \left| \llbracket w_i \rrbracket \right|^{p} \right) dA \right]^{1/p}. \label{norm_standard}
    \end{align}
\end{definition}

\begin{definition}[Full Symmetric Gradient Norm] \label{norm_def_versatile}
    Consider the vector-valued function $\wbold \in \bm{W}^{1,p} \left(\mathcal{T}_h\right)$. Then,
    \begin{align}
        \nonumber \left\| \wbold \right\|_{\text{sym},p} &= \left[ \left\| \nabla_h \wbold + \nabla_h \wbold^T - \frac{2}{3} \left(\nabla_h \cdot \wbold\right) \mathbb{I} \right\|_{\bm{L}^p\left(\Omega\right) \times \bm{L}^p\left(\Omega\right)}^{p} + \sum_{F \in \mathcal{F}_h} \frac{1}{h_{F}^{p-1}} \left\| \llbracket \wbold \rrbracket \right\|_{\bm{L}^{p} \left(F\right)}^{p} \right]^{1/p} \\[1.5ex]
        \nonumber & = \Bigg[ \sum_{K \in \mathcal{T}_h} \int_{K} \left( \sum_{i,j}^d \left| \partial_{j} w_i + \partial_{i} w_j - \frac{2}{3} \left(\sum_{k}^{d} \partial_k w_k \right) \delta_{ij} \right|^p \right) dV 
        \\[1.5ex]
        &+ \sum_{F \in \mathcal{F}_h} \frac{1}{h_{F}^{p-1}} \int_{F} \left( \sum_{i}^{d} \left| \llbracket w_i \rrbracket \right|^{p} \right) dA \Bigg]^{1/p}, \label{norm_versatile}
    \end{align}
    is a norm on $\Omega$. 
\end{definition}

\subsection{Analysis of Bilinear and Trilinear Forms}

\begin{lemma}[Coercivity of the Viscous Bilinear Form]
Suppose we choose a generic test function $r_h \in R_h$, and we assume that $d = 2$ or 3. Furthermore, we choose $\varepsilon > C_{\text{tr},2}^2N_{\partial}$, where $C_{\text{tr},2}$ and $N_{\partial}$ are constants which depend on the mesh topology. Then, the bilinear form $\underline{a}_h$ in Eq.~\eqref{diff_bilinear_temp} is coercive on $R_h$, such that
\begin{align}
    \forall r_h \in R_h,\qquad \underline{a}_h \left( r_h, r_h \right) \geq C_{I} \left\| r_h \right\|_{\text{grad},2}^2, \label{coercive_res}
\end{align}
where $C_{I} =  \left(\varepsilon - C_{\text{tr},2}^2 N_{\partial}\right)/\left(1+\varepsilon\right)$ is a positive constant independent of $h$.
\label{coercive_vis_temp_lemma}
\end{lemma}

\begin{proof}
    The proof appears in~\cite{DiPietro11} Lemma 4.12 (p.~129).
\end{proof}

\begin{lemma}[Semi-Coercivity of the Convective Trilinear Form]
Consider test functions $\bm{\xi}_h \in \bm{W}_h$ and $r_h \in R_h$. Then, the trilinear form $\underline{c}_h$ in Eq.~\eqref{trilinear_temp} is semi-coercive on $\bm{W}_h \times R_h$, such that
\begin{align}
    \forall \left(\bm{\xi}_h, r_h \right) \in \bm{W}_h \times R_h, \qquad \underline{c}_h \left(\bm{\xi}_h; r_h, r_h \right) = \left| r_h \right|_{\bm{\xi}_h}^2, \label{convective_coercive_temp}
\end{align}
where
\begin{align}
    \left| r_h \right|_{\bm{\xi}_h} = \left( \delta \iipbf{\left| \bm{\xi}_h \cdot \nbold_F \right|\llbracket r_h \rrbracket}{\llbracket r_h \rrbracket} \right)^{1/2},
\end{align}
is a seminorm on $\Omega$.
\label{coercive_tri_temp_lemma}
\end{lemma}

\begin{proof}
    One may begin by substituting $\theta_h = r_h$ into Eq.~\eqref{trilinear_temp} as follows
    \begin{align}
    \underline{c}_h \left(\bm{\xi}_h; r_h, r_h \right) &= \ipt{\bm{\xi}_h \cdot \nabla_h r_h}{r_h} + \frac{1}{2} \ipt{\left(\nabla \cdot \bm{\xi}_h \right) r_h}{r_h} \label{temp_coercive_one} \\[1.5ex]
    & \nonumber - \iipbf{ \left( \bm{\xi}_h \cdot \nbold_F \right) \llbracket r_h \rrbracket}{\llcurve r_h \rrcurve} + \delta \iipbf{\left| \bm{\xi}_h \cdot \nbold_F \right|\llbracket r_h \rrbracket}{\llbracket r_h \rrbracket}. 
    \end{align}
    Next, we note that the following identity holds
    \begin{align*}
        \ipt{\bm{\xi}_h \cdot \nabla_h r_h}{r_h} + \frac{1}{2} \ipt{\left(\nabla \cdot \bm{\xi}_h \right) r_h}{r_h} =  \iipbf{ \left(\bm{\xi}_h \cdot \nbold_F \right) \llbracket r_h \rrbracket}{\llcurve r_h \rrcurve}.
    \end{align*}
    Upon substituting this identity into Eq.~\eqref{temp_coercive_one}, one obtains
    \begin{align}
        \underline{c}_h \left(\bm{\xi}_h; r_h, r_h \right) = \delta \iipbf{\left| \bm{\xi}_h \cdot \nbold_F \right|\llbracket r_h \rrbracket}{\llbracket r_h \rrbracket}. \label{temp_coercive_two}
    \end{align}
    Finally, on substituting the definition of the seminorm into Eq.~\eqref{temp_coercive_two}, we obtain the desired result (Eq.~\eqref{convective_coercive_temp}).
\end{proof}

\subsection{Analysis of Discrete Stability}

\begin{theorem}[Stability of the Discrete Temperature] Consider the mixed finite element methods in Eqs.~\eqref{mass_cons_disc} -- \eqref{energy_cons_disc}, in conjunction with a forcing function $\widetilde{f}_{T}$ where $\widetilde{f}_{T} \left(t\right) \in L^1 \left(t_0, t_n; L^2 \left(\Omega \right)\right)$, a discrete pressure field $\widetilde{p}_h \in Q_h$ where $\widetilde{p}_h \left(t\right) \in L^2 \left(t_0, t_n; L^4 \left(\Omega \right) \right)$, and a discrete velocity field $\ubold_h  \in \bm{W}_h$ where $\ubold_h \left(t\right) \in L^2 \left(t_0, t_n; \bm{W}^{1,4} \left(\mathcal{T}_h\right) \right)$. Subject to these assumptions, the discrete temperature $T_h \in R_h$ is governed by the following equation at time $t_n \geq t_0$
\begin{align}
    \nonumber & \frac{1}{2} \left\| T_h \left(t_n\right) \right\|_{L^2 \left(\Omega\right)}^2 + \left| T_h \right|_{L^2 \left(t_0, t_n; \ubold_h \right)}^2 + \alpha_h \gamma_h C_{I} \left\| T_h \right\|_{L^2 \left(t_0, t_n; \text{grad}, 2 \right)}^2 
    \\[1.5ex]
    \nonumber &\leq \frac{1}{2} \Bigg( 4 \left\| T_h \left(t_0\right) \right\|_{L^2 \left(\Omega\right)}^2 + 7  C_{II} \frac{\nu_h^2}{C_v^2} \left\| \ubold_h \right\|_{L^2 \left(t_0, t_n; \text{grad},4 \right)}^4 
    \\[1.5ex]
    & + 7 \frac{1}{C_v^2} \left\| \widetilde{p}_h \right\|_{L^2 \left(t_0, t_n; L^4 \left(\Omega\right) \right)}^{2} \left\| \nabla \cdot \ubold_h \right\|_{L^2 \left(t_0, t_n; L^4 \left(\Omega\right) \right)}^{2} + 7 \left\| \widetilde{f}_{T} \right\|_{L^1 \left(t_0, t_n; L^2 \left(\Omega\right) \right)}^2 \Bigg),
\label{temperature_l2_bound}
\end{align}
where $C_{I}$ and $C_{II}$ are constants that are independent of $h$, and
\begin{align}
    & \left| T_h \right|_{L^2 \left(t_0, t_n; \ubold_h \right)} = \left(\int_{t_0}^{t_n} \left| T_h \left(s\right) \right|_{\ubold_h}^2 ds \right)^{1/2}, \label{temp_norm_one}
    \\[1.5ex]
    & \left\| T_h \right\|_{L^2 \left(t_0, t_n; \text{grad}, 2 \right)} = \left(\int_{t_0}^{t_n} \left\| T_h \left(s\right) \right\|_{\text{grad},2}^2 ds \right)^{1/2}, 
    \\[1.5ex]
    & \left\| \ubold_h \right\|_{L^2 \left(t_0, t_n; \text{grad}, 4 \right)} = \left(\int_{t_0}^{t_n} \left\| \ubold_h \left(s\right) \right\|_{\text{grad},4}^2 ds \right)^{1/2},  
    \\[1.5ex]
    &\left\| \widetilde{f}_{T} \right\|_{L^1 \left(t_0, t_n; L^2 \left(\Omega\right) \right)} = \int_{t_0}^{t_n} \left\| \widetilde{f}_{T} \left(s\right) \right\|_{L^2 \left(\Omega\right)} ds, \label{temp_norm_four}
\end{align}
are seminorms and norms on $\left(t_0, t_n\right) \times \Omega$. \label{temperature_theorem}
\end{theorem}

\begin{proof}
    We start by setting $r_h = T_h$ in Eq.~\eqref{incomp_form_one_C} as follows
    \begin{align*}
        \ipt{\partial_t T_h}{T_h} + \underline{c}_h\left(\ubold_h; T_h, T_h\right) + \alpha_h \gamma_h \, \underline{a}_h \left(T_h, T_h \right) = \ipt{\Phi \left(\ubold_h\right) + \Psi\left(\ubold_h, \widetilde{p}_h\right) + \widetilde{f}_{T}}{T_h},
    \end{align*}
    or equivalently
    \begin{align*}
        \frac{1}{2} \frac{d}{dt} \left\| T_h \right\|_{L^2 \left(\Omega\right)}^2 + \underline{c}_h\left(\ubold_h; T_h, T_h\right) + \alpha_h \gamma_h \, \underline{a}_h \left(T_h, T_h \right) = \ipt{\Phi \left(\ubold_h\right) + \Psi\left(\ubold_h, \widetilde{p}_h\right) + \widetilde{f}_{T}}{T_h}.
    \end{align*}
    Next, we invoke the coercivity of $\underline{a}_h$ (Lemma~\ref{coercive_vis_temp_lemma}) and the semi-coercivity of $\underline{c}_h$ (Lemma~\ref{coercive_tri_temp_lemma}) as follows
    \begin{align}
        \frac{1}{2} \frac{d}{dt} \left\| T_h \right\|_{L^2 \left(\Omega\right)}^2 +  \left| T_h \right|_{\ubold_h}^2 + \alpha_h \gamma_h C_{I} \left\| T_h \right\|_{\text{grad},2}^2  \leq \ipt{\Phi \left(\ubold_h\right) + \Psi\left(\ubold_h, \widetilde{p}_h\right) + \widetilde{f}_{T}}{T_h}.
        \label{temp_bound_zero}
    \end{align}
    Based on this equation, we observe that
    \begin{align*}
        \frac{1}{2} \frac{d}{dt} \left\| T_h \right\|_{L^2 \left(\Omega\right)}^2 \leq \ipt{\Phi \left(\ubold_h\right) + \Psi\left(\ubold_h, \widetilde{p}_h\right) + \widetilde{f}_{T}}{T_h},
    \end{align*}
    and equivalently, by the Cauchy-Schwarz and Triangle inequalities
    \begin{align}
         \nonumber \left\| T_h \right\|_{L^2 \left(\Omega\right)} \frac{d}{dt} \left\| T_h \right\|_{L^2 \left(\Omega\right)} & \leq \left( \left\| \Phi \left(\ubold_h\right) \right\|_{L^2 \left(\Omega\right)} + \left\| \Psi\left(\ubold_h, \widetilde{p}_h\right) \right\|_{L^2 \left(\Omega\right)} + \left\| \widetilde{f}_{T} \right\|_{L^2\left(\Omega\right)} \right) \left\| T_h \right\|_{L^2 \left(\Omega\right)} \\[1.5ex]
         \nonumber \frac{d}{dt} \left\| T_h \right\|_{L^2 \left(\Omega\right)} & \leq \left\| \Phi \left(\ubold_h\right) \right\|_{L^2 \left(\Omega\right)} + \left\| \Psi\left(\ubold_h, \widetilde{p}_h\right) \right\|_{L^2 \left(\Omega\right)} + \left\| \widetilde{f}_{T} \right\|_{L^2\left(\Omega\right)}.
    \end{align}
    Upon expanding the RHS of this expression, and using the Triangle and Cauchy-Schwarz inequalities again, we obtain
    \begin{align}
         \nonumber \frac{d}{dt} \left\| T_h \right\|_{L^2 \left(\Omega\right)}  & \leq \frac{\nu_h}{C_v} \left(  \left\| \nabla_h \ubold_h:\nabla_h \ubold_h \right\|_{L^2 \left(\Omega\right)} + \left\| \nabla_h \ubold_{h}^{T}:\nabla_h \ubold_h \right\|_{L^2 \left(\Omega\right)} + \frac{2}{3} \left\| \left( \nabla \cdot \ubold_h \right)^2 \right\|_{L^2 \left(\Omega\right)}  \right)
         \\[1.5ex]
         \nonumber & + \frac{1}{C_v} \left\| \widetilde{p}_h \left( \nabla \cdot \ubold_h \right) \right\|_{L^2\left(\Omega\right)}   + \left\| \widetilde{f}_{T} \right\|_{L^2\left(\Omega\right)} 
         \\[1.5ex]
          \nonumber & \leq \frac{\nu_h}{C_v} \left(  \left\| \nabla_h \ubold_h:\nabla_h \ubold_h \right\|_{L^2 \left(\Omega\right)} + \left\| \nabla_h \ubold_{h}^{T}:\nabla_h \ubold_h \right\|_{L^2 \left(\Omega\right)} + \frac{2}{3} \left\| \nabla \cdot \ubold_h  \right\|_{L^4 \left(\Omega\right)}^{2}  \right)
         \\[1.5ex]
         \nonumber & + \frac{1}{C_v} \left\| \widetilde{p}_h  \right\|_{L^4\left(\Omega\right)} \left\| \nabla \cdot \ubold_h \right\|_{L^4\left(\Omega\right)} + \left\| \widetilde{f}_{T} \right\|_{L^2\left(\Omega\right)} 
         \\[1.5ex]
         & \leq 2d \left(1 + \frac{\sqrt{d}}{3} \right) \frac{\nu_h}{C_v} \left\| \ubold_h \right\|_{\text{grad},4}^2 + \frac{1}{C_v} \left\| \widetilde{p}_h  \right\|_{L^4\left(\Omega\right)} \left\| \nabla \cdot \ubold_h \right\|_{L^4\left(\Omega\right)} + \left\| \widetilde{f}_{T} \right\|_{L^2\left(\Omega\right)}.\label{temp_bound_one}
    \end{align}
    Note: on the last line we have used the broken norm inequalities from Lemma~\ref{grad_inequality_lemma} in the Appendix. Next, we integrate Eq.~\eqref{temp_bound_one} from $t = t_0$ to $t = t_n$ as follows
    \begin{align}
        \nonumber & \left\| T_h \left(t_n\right) \right\|_{L^2 \left(\Omega\right)} \\[1.5ex]
        \nonumber & \leq \left\| T_h \left(t_0\right) \right\|_{L^2 \left(\Omega\right)} +  2d \left(1 + \frac{\sqrt{d}}{3} \right) \frac{\nu_h}{C_v} \int_{t_0}^{t_n} \left\| \ubold_h \left(s\right) \right\|_{\text{grad},4}^2 ds 
        \\[1.5ex]
        \nonumber & + \frac{1}{C_v} \int_{t_0}^{t_n} \left\| \widetilde{p}_h \left(s\right) \right\|_{L^4\left(\Omega\right)} \left\| \nabla \cdot \ubold_h \left(s\right) \right\|_{L^4\left(\Omega\right)} ds + \int_{t_0}^{t_n} \left\| \widetilde{f}_{T} \left(s\right) \right\|_{L^2\left(\Omega\right)} ds
        \\[1.5ex]
        \nonumber & = \left\| T_h \left(t_0\right) \right\|_{L^2 \left(\Omega\right)} +  2d \left(1 + \frac{\sqrt{d}}{3} \right) \frac{\nu_h}{C_v} \left\| \ubold_h \right\|_{L^2 \left(t_0, t_n; \text{grad},4 \right)}^2 
        \\[1.5ex]
        \nonumber &+ \frac{1}{C_v} \int_{t_0}^{t_n} \left\| \widetilde{p}_h \left(s\right) \right\|_{L^4\left(\Omega\right)} \left\| \nabla \cdot \ubold_h \left(s\right) \right\|_{L^4\left(\Omega\right)} ds + \left\| \widetilde{f}_{T} \right\|_{L^1 \left(t_0, t_n; L^2 \left(\Omega\right) \right)},
    \end{align}
    or equivalently, after applying Holder's inequality
    \begin{align}
        \nonumber & \left\| T_h \left(t_n\right) \right\|_{L^2 \left(\Omega\right)} 
        \\[1.5ex]
        \nonumber & \leq \left\| T_h \left(t_0\right) \right\|_{L^2 \left(\Omega\right)} +  2d \left(1 + \frac{\sqrt{d}}{3} \right) \frac{\nu_h}{C_v} \left\| \ubold_h \right\|_{L^2 \left(t_0, t_n; \text{grad},4 \right)}^2 
        \\[1.5ex]
        \nonumber &+ \frac{1}{C_v} \left( \int_{t_0}^{t_n} \left\| \widetilde{p}_h \left(s\right) \right\|_{L^4\left(\Omega\right)}^{2} ds \right)^{1/2} \left( \int_{t_0}^{t_n} \left\| \nabla \cdot \ubold_h \left(s\right) \right\|_{L^4\left(\Omega\right)}^{2} ds \right)^{1/2} + \left\| \widetilde{f}_{T} \right\|_{L^1 \left(t_0, t_n; L^2 \left(\Omega\right) \right)}
        \\[1.5ex]
        \nonumber & \leq \left\| T_h \left(t_0\right) \right\|_{L^2 \left(\Omega\right)} +  2d \left(1 + \frac{\sqrt{d}}{3} \right) \frac{\nu_h}{C_v} \left\| \ubold_h \right\|_{L^2 \left(t_0, t_n; \text{grad},4 \right)}^2 
        \\[1.5ex]
        &+ \frac{1}{C_v} \left\| \widetilde{p}_h \right\|_{L^2 \left(t_0, t_n; L^4 \left(\Omega\right) \right)} \left\| \nabla \cdot \ubold_h \right\|_{L^2 \left(t_0, t_n; L^4 \left(\Omega\right) \right)}  + \left\| \widetilde{f}_{T} \right\|_{L^1 \left(t_0, t_n; L^2 \left(\Omega\right) \right)}. \label{temp_bound_two}
    \end{align}
    We will utilize this result shortly. For now, we turn our attention back to Eq.~\eqref{temp_bound_zero}. On integrating this equation from $t = t_0$ to $t = t_n$, we find that
    \begin{align}
        \nonumber & \frac{1}{2} \left\| T_h \left(t_n\right) \right\|_{L^2 \left(\Omega\right)}^2 + \int_{t_0}^{t_n} \left( \left| T_h \left(s\right) \right|_{\ubold_h}^2 + \alpha_h \gamma_h C_{I} \left\| T_h \left(s\right) \right\|_{\text{grad},2}^2 \right) ds \\[1.5ex]
        &\leq \frac{1}{2} \left\| T_h \left(t_0\right) \right\|_{L^2 \left(\Omega\right)}^2 + \int_{t_0}^{t_n} \ipt{\Phi \left(\ubold_h \left(s\right) \right) + \Psi\left(\ubold_h \left(s\right), \widetilde{p}_h \left(s\right) \right) + \widetilde{f}_{T} \left(s\right)}{T_h \left(s\right)} ds. \label{temp_bound_three}
    \end{align}
    We can rewrite the last term on the RHS of Eq.~\eqref{temp_bound_three} as follows
    \begin{align}
        \nonumber & \int_{t_0}^{t_n} \ipt{\Phi \left( \ubold_h  \left(s\right) \right) + \Psi\left(\ubold_h \left(s\right), \widetilde{p}_h \left(s\right) \right) + \widetilde{f}_{T} \left(s\right)}{T_h \left(s\right)} ds 
        \\[1.5ex]
        \nonumber & \leq \int_{t_0}^{t_n} \left[ \left(\left\| \Phi \left( \ubold_h  \left(s\right) \right) \right\|_{L^2\left(\Omega\right)} + \left\| \Psi\left(\ubold_h \left(s\right), \widetilde{p}_h \left(s\right) \right) \right\|_{L^2 \left(\Omega\right)} + \left\| \widetilde{f}_{T} \left(s\right) \right\|_{L^2\left(\Omega\right)} \right) \left\| T_h \left(s\right)\right\|_{L^2\left(\Omega\right)} \right] ds 
        \\[1.5ex]
        \nonumber & \leq \int_{t_0}^{t_n} \Bigg[ \left(\left\| \Phi \left( \ubold_h  \left(s\right) \right) \right\|_{L^2\left(\Omega\right)} + \left\| \Psi\left(\ubold_h \left(s\right), \widetilde{p}_h \left(s\right) \right) \right\|_{L^2 \left(\Omega\right)} + \left\| \widetilde{f}_{T} \left(s\right) \right\|_{L^2\left(\Omega\right)} \right) 
        \\[1.5ex]
        \nonumber & \times \Bigg(\left\| T_h \left(t_0\right) \right\|_{L^2 \left(\Omega\right)} + 2d \left(1 + \frac{\sqrt{d}}{3} \right) \frac{\nu_h}{C_v} \left\| \ubold_h \right\|_{L^2 \left(t_0, s; \text{grad},4 \right)}^2 
        \\[1.5ex]
        \nonumber &+ \frac{1}{C_v} \left\| \widetilde{p}_h \right\|_{L^2 \left(t_0, s; L^4 \left(\Omega\right) \right)} \left\| \nabla \cdot \ubold_h \right\|_{L^2 \left(t_0, s; L^4 \left(\Omega\right) \right)} + \left\| \widetilde{f}_{T} \right\|_{L^1 \left(t_0, s; L^2 \left(\Omega\right) \right)} \Bigg)  \Bigg] ds,
    \end{align}
    and furthermore
    \begin{align}    
        \nonumber & \int_{t_0}^{t_n} \ipt{\Phi \left( \ubold_h  \left(s\right) \right) + \Psi\left(\ubold_h \left(s\right), \widetilde{p}_h \left(s\right) \right) + \widetilde{f}_{T} \left(s\right)}{T_h \left(s\right)} ds 
        \\[1.5ex]
        \nonumber & \leq \int_{t_0}^{t_n} \left(\left\| \Phi \left( \ubold_h  \left(s\right) \right) \right\|_{L^2\left(\Omega\right)} + \left\| \Psi\left(\ubold_h \left(s\right), \widetilde{p}_h \left(s\right) \right) \right\|_{L^2 \left(\Omega\right)} + \left\| \widetilde{f}_{T} \left(s\right) \right\|_{L^2\left(\Omega\right)} \right) ds 
        \\[1.5ex]
        \nonumber & \times \Bigg(\left\| T_h \left(t_0\right) \right\|_{L^2 \left(\Omega\right)} + 2d \left(1 + \frac{\sqrt{d}}{3} \right) \frac{\nu_h}{C_v} \left\| \ubold_h \right\|_{L^2 \left(t_0, t_n; \text{grad},4 \right)}^2 
        \\[1.5ex]
        \nonumber &+ \frac{1}{C_v} \left\| \widetilde{p}_h \right\|_{L^2 \left(t_0, t_n; L^4 \left(\Omega\right) \right)} \left\| \nabla \cdot \ubold_h \right\|_{L^2 \left(t_0, t_n; L^4 \left(\Omega\right) \right)} + \left\| \widetilde{f}_{T} \right\|_{L^1 \left(t_0, t_n; L^2 \left(\Omega\right) \right)} \Bigg).
    \end{align}
    Here, we have used the Cauchy-Schwarz inequality, the Triangle inequality, and Eq.~\eqref{temp_bound_two}. Next, we bound the remaining terms in the integrand above (employing the same techniques that we used to derive Eq.~\eqref{temp_bound_two}), and we obtain
    \begin{align}
        \nonumber & \int_{t_0}^{t_n} \ipt{\Phi \left( \ubold_h  \left(s\right) \right) + \Psi\left(\ubold_h \left(s\right), \widetilde{p}_h \left(s\right) \right) + \widetilde{f}_{T} \left(s\right)}{T_h \left(s\right)} ds 
        \\[1.5ex]
        \nonumber & \leq \Bigg( 2d \left(1 + \frac{\sqrt{d}}{3} \right) \frac{\nu_h}{C_v} \left\| \ubold_h \right\|_{L^2 \left(t_0, t_n; \text{grad},4 \right)}^2 
        \\[1.5ex]
        \nonumber &+ \frac{1}{C_v} \left\| \widetilde{p}_h \right\|_{L^2 \left(t_0, t_n; L^4 \left(\Omega\right) \right)} \left\| \nabla \cdot \ubold_h \right\|_{L^2 \left(t_0, t_n; L^4 \left(\Omega\right) \right)}  + \left\| \widetilde{f}_{T} \right\|_{L^1 \left(t_0, t_n; L^2 \left(\Omega\right) \right)} \Bigg)
        \\[1.5ex]
        \nonumber & \times \Bigg( \left\| T_h \left(t_0\right) \right\|_{L^2 \left(\Omega\right)} + 2d \left(1 + \frac{\sqrt{d}}{3} \right) \frac{\nu_h}{C_v} \left\| \ubold_h \right\|_{L^2 \left(t_0, t_n; \text{grad},4 \right)}^2 
        \\[1.5ex]
        \nonumber &+ \frac{1}{C_v} \left\| \widetilde{p}_h \right\|_{L^2 \left(t_0, t_n; L^4 \left(\Omega\right) \right)} \left\| \nabla \cdot \ubold_h \right\|_{L^2 \left(t_0, t_n; L^4 \left(\Omega\right) \right)} + \left\| \widetilde{f}_{T} \right\|_{L^1 \left(t_0, t_n; L^2 \left(\Omega\right) \right)} \Bigg)
        \\[1.5ex]
        \nonumber & \leq \frac{3}{2} \left\| T_h \left(t_0\right) \right\|_{L^2 \left(\Omega\right)}^2 + \frac{7}{2} C_{II} \frac{\nu_h^2}{C_v^2} \left\| \ubold_h \right\|_{L^2 \left(t_0, t_n; \text{grad},4 \right)}^4 
        \\[1.5ex]
        & + \frac{7}{2} \frac{1}{C_v^2} \left\| \widetilde{p}_h \right\|_{L^2 \left(t_0, t_n; L^4 \left(\Omega\right) \right)}^{2} \left\| \nabla \cdot \ubold_h \right\|_{L^2 \left(t_0, t_n; L^4 \left(\Omega\right) \right)}^{2} + \frac{7}{2} \left\| \widetilde{f}_{T} \right\|_{L^1 \left(t_0, t_n; L^2 \left(\Omega\right) \right)}^2, \label{temp_bound_four}
    \end{align}
    where $C_{II} = 4 d^{2} \left(1 + \frac{\sqrt{d}}{3} \right)^2$. Finally, upon combining Eq.~\eqref{temp_bound_four} with Eq.~\eqref{temp_bound_three}, and substituting in the space-time norm definitions from Eqs.~\eqref{temp_norm_one}--\eqref{temp_norm_four}, we obtain the desired result (see Eq.~\eqref{temperature_l2_bound}).
\end{proof}

\begin{corollary} [Pointwise Divergence-Free Case]
    Suppose that the mixed finite element methods in Eqs.~\eqref{mass_cons_disc} -- \eqref{energy_cons_disc} are pointwise divergence-free. In addition, suppose we impose a forcing function $\widetilde{f}_{T}$ where $\widetilde{f}_{T} \left(t\right) \in L^1 \left(t_0, t_n; L^2 \left(\Omega \right)\right)$ and a discrete velocity field $\ubold_h  \in \bm{W}_h$ where $\ubold_h \left(t\right) \in L^2 \left(t_0, t_n; \bm{W}^{1,4} \left(\mathcal{T}_h\right) \right)$. Subject to these assumptions, the discrete temperature $T_h \in R_h$ is governed by the following equation at time $t_n \geq t_0$
    \begin{align}
        \nonumber & \frac{1}{2} \left\| T_h \left(t_n\right) \right\|_{L^2 \left(\Omega\right)}^2 + \left| T_h \right|_{L^2 \left(t_0, t_n; \ubold_h \right)}^2 + \alpha_h \gamma_h C_{I} \left\| T_h \right\|_{L^2 \left(t_0, t_n; \text{grad}, 2 \right)}^2 
        \\[1.5ex]
        &\leq \frac{1}{2} \Bigg( 4 \left\| T_h \left(t_0\right) \right\|_{L^2 \left(\Omega\right)}^2 + 7  C_{II} \frac{\nu_h^2}{C_v^2} \left\| \ubold_h \right\|_{L^2 \left(t_0, t_n; \text{grad},4 \right)}^4 + 7 \left\| \widetilde{f}_{T} \right\|_{L^1 \left(t_0, t_n; L^2 \left(\Omega\right) \right)}^2 \Bigg).
        \label{temperature_l2_bound_corollary}
    \end{align}
\label{temperature_corollary}
\end{corollary}

\begin{proof}
    The proof immediately follows from setting $\nabla \cdot \ubold_h = 0$ pointwise in Theorem~\ref{temperature_theorem}.
\end{proof}

\begin{theorem}[Stability of the Discrete Velocity Field] Consider the mixed finite element methods in Eqs.~\eqref{mass_cons_disc} -- \eqref{energy_cons_disc}, in conjunction with a forcing function $\widetilde{\bm{f}}_{\ubold} \in L^1 \left(t_0, t_n; \bm{L}^2 \left(\Omega \right)\right)$ and an initial condition $\ubold_h \left(t_0 \right) \in \bm{W}_h \subset \bm{H}_{0}(\text{div};\Omega)$. Subject to these assumptions, the velocity field is governed by the following equation at time $t_n \geq t_0$
\begin{align}
        &\nonumber \frac{1}{2} \left\| \ubold_h \left(t_n \right) \right\|_{\bm{L}^2 \left(\Omega\right)}^2 + \left| \ubold_h \right|_{L^2 \left(t_0, t_n; \ubold_h \right)}^2 + C_{III} \nu_{h}  \left\| \ubold_h \right\|_{L^2 \left(t_0, t_n; \text{sym}, 2 \right)}^2
        \\[1.5ex]
       & \leq \frac{1}{2} \left( 3 \left\| \ubold_h \left(t_0 \right) \right\|_{\bm{L}^2 \left(\Omega\right)}^2 + 5 \beta_{h}^{2} g^{2} \left\| T_h \right\|_{L^1 \left(t_0, t_n; L^2 \left(\Omega\right) \right)}^2 + 5 \left\| \widetilde{\bm{f}}_{\ubold} \right\|_{L^1 \left(t_0, t_n; \bm{L}^2 \left(\Omega\right) \right)}^2 \right),  \label{velocity_l2_bound}
\end{align}
where $C_{III}$ is a constant independent of $h$, and
\begin{align}
    & \left\| \ubold_h \right\|_{L^2 \left(t_0, t_n; \text{sym}, 2 \right)} = \left(\int_{t_0}^{t_n} \left\| \ubold_h \left(s\right) \right\|_{\text{sym},2}^2 ds \right)^{1/2},  
\end{align}
is a norm on $\left(t_0, t_n\right) \times \Omega$.
\label{velocity_theorem}
\end{theorem}

\begin{proof}
    The important aspects of the proof are standard, and follow the arguments in Lemma 3.1 of~\cite{arndt2015local} and Theorem 7.1 of~\cite{chen2020versatile}. 
\end{proof}

\begin{remark}
    The stability of the discrete temperature field depends on the stability of the discrete velocity field, as shown by Theorem~\ref{temperature_theorem} and Corollary~\ref{temperature_corollary}. Conversely, the stability of the discrete velocity field depends on the stability of the discrete temperature field, as shown by Theorem~\ref{velocity_theorem}. Therefore, it is difficult to establish independent stability of either field, and (theoretically speaking) this  may result in undesirable interference between the two fields. Fortunately, in most cases the coupling between fields is weak as one of the following assumptions holds true:
    \begin{itemize}
        \item The buoyancy term (with coefficient $5 \beta_h^2 g^2$) on the right hand side of Eq.~\eqref{velocity_l2_bound} is small.
        \item The viscous dissipation term (with coefficient $7 C_{II} \nu_{h}^{2}/C_{v}^{2}$) on the right hand side of Eq.~\eqref{temperature_l2_bound} or Eq.~\eqref{temperature_l2_bound_corollary} is small. 
    \end{itemize}
\end{remark}

\section{Numerical Experiments}
In this section, the results of several numerical simulations are presented to demonstrate the performance of the proposed methods. The following simulations were all performed using Taylor-Hood elements with polynomials of degree $k$, $k+1$, and $k+1$ for the pressure, temperature, and velocity spaces respectively; i.e.~for cases with $k = 1$ the polynomials for each space were degree 1, 2, and 2 respectively. In addition, we imposed a zero integral mean condition for the pressure via a Lagrange multiplier. The convective numerical fluxes were computed using upwind biased fluxes with $\zeta = \delta = 0.5$, and the viscous numerical fluxes were computed using $\eta = \varepsilon = 3(k+1)(k+2)$. In each case, a high-order BDF3 scheme was used for the time discretization. The meshes were developed using rectangular grids where the quadrilateral elements were split along the diagonals to create triangles. Throughout this section, mesh dimensions are reported as $N \times M$, where $N$ and $M$ refer to the number of quadrilaterals in the $x$ and $y$ directions, respectively. The total number of elements for each case was $2N \times 2M$ due to the splitting mentioned previously. Finally, each simulation was performed in the open-source finite element software package FEniCS~\cite{alnaes2015fenics}. 

\subsection{Order of Accuracy Test}
For the first test case, we compared solutions from our method to an exact solution in order to check the convergence rate. To this end, we considered the traveling temperature wave proposed by~\cite{dallmann2015finite}, which can be defined on the rectangular domain $\Omega = \left[-0.5,1.5\right]
\times\left[0,1\right]$ as follows
\begin{align*}
    \ubold &= \left(100, 0\right), \qquad
    \widetilde{p} = 1, \\[1.5ex]
    T &= \frac{1}{\sqrt{1+3200\alpha t}} \exp \left[ - \frac{200\left(1 + 200t - 2x\right)^2}{1+3200\alpha t} \right],
\end{align*}
for $t\in[0,0.005]$. We also defined the gravitational acceleration and forcing functions as follows 
\begin{align*}
    \gbold &= \left(0, -1\right), \qquad \widetilde{f}_{T} = 0, \\[1.5ex]
    \widetilde{\bm{f}}_{u} &= \left(0, -\beta \frac{1}{\sqrt{1+3200\alpha t}} \exp \left[ - \frac{200\left(1 + 200t - 2x\right)^2}{1+3200\alpha t} \right] \right).
\end{align*}
Here, we considered a dimensionless formulation with $\alpha = \beta = \gamma = \nu = \rho = 1$. In subsequent cases, a dimensional formulation was considered. 

For this case, a temperature peak was initially located at $x = \frac{1}{2}$ at $t = 0$ and moved to $x = 1$ at the final time $t = 0.005$. We compared our results to the exact solution at the final time. The time step for all polynomial degrees considered was $\Delta t = 1\times 10^{-5}$. Periodic boundary conditions were applied at the domain boundaries. The meshes were uniform and consisted of $N \times \frac{N}{2}$ elements. On these meshes, we utilized Taylor-Hood spaces with degrees $k = 1,2,$ and $3$. For the case of $k = 1$, mesh resolutions of $N = 16, 32, 64,$ and $128$ were used, and for $k = 2$ and $3$ mesh resolutions of $N = 4, 8, 16,$ and $32$ were used. We expected the discrete temperature to converge at a rate of $k + 2$ since the associated polynomial space was degree $k + 1$.   
\begin{table}[h]
\centering
\begin{tabular}{|l|l|l|l|l|}
\hline
$k$                  & $h$       & dofs    & $L^2$ error & order  \\ \hline
\multirow{4}{*}{1} & 0.17677 & 1777   & 1.8287e-3  & -      \\ \cline{2-5} 
                   & 0.08838 & 6881   & 2.3516e-4  & 2.9591 \\ \cline{2-5} 
                   & 0.04419 & 27073  & 2.4724e-5  & 3.2496  \\ \cline{2-5} 
                   & 0.02209 & 107393 & 2.8976e-6  & 3.0930 \\ \hline
\multirow{4}{*}{2} & 0.70710 & 293    & 2.4046e-2  & -      \\ \cline{2-5} 
                   & 0.35355 & 1081   & 4.0106e-3  & 2.5839 \\ \cline{2-5} 
                   & 0.17677 & 4145   & 1.9883e-4  & 4.3342 \\ \cline{2-5} 
                   & 0.08838 & 16225  & 1.0639e-5  & 4.2241 \\ \hline
\multirow{4}{*}{3} & 0.70710 & 517    & 2.0978e-2  & -      \\ \cline{2-5} 
                   & 0.35355 & 1945   & 2.1626e-3  & 3.2780 \\ \cline{2-5} 
                   & 0.17677 & 7537   & 2.5591e-5  & 6.4010 \\ \cline{2-5} 
                   & 0.08838 & 29665  & 7.6643e-7  & 5.0613 \\ \hline
\end{tabular}
\caption{Temperature $L^2$ error for various polynomial degrees $k$ and maximum element diameters $h$. }
\label{Conv_tab}
\end{table}
We see from table~\ref{Conv_tab} that we recovered the expected orders of accuracy. 

\subsection{Heated Cavity Test}
The second test case was a heated cavity as described by~\cite{dallmann2015finite}. This case consisted of a square cavity $\Omega = \left[0,1\right]^2$ with stationary walls. The flow was driven by a temperature difference between the left and right walls, and thus consisted of purely natural convection. Gravity $\bold{g}=(0,-1)^T$ $\frac{m}{s^2}$ in conjunction with buoyancy effects influenced the fluid motion. For all heated cavity simulations, a fixed Prandlt number $Pr = 0.71$ defined as
\begin{align*}
    Pr = \frac{\nu}{\alpha}, 
\end{align*}
was used. Fluid properties for all cases were $\alpha = 2.208 \times 10^{-5}$ $\frac{m^2}{s}$ , $C_v = 717.8$ $\frac{J}{kg-K}$, $\gamma = 1$, $\rho=1$ $\frac{kg}{m^3}$ and $\nu = 1.568 \times 10^{-5}$ $\frac{m^2}{s}$ which denote an air-like fluid. All walls were equipped with no-slip boundary conditions, where the left and right walls had fixed Dirichlet temperature boundary conditions $T_{left} = 0.5$ K and $T_{right} = -0.5$ K, and where the top and bottom walls were adiabatic. For this set of simulations the Rayleigh number $Ra$ was varied throughout. Specifically, we decided to vary the Rayleigh number by varying the parameter $\beta$, using the following formulas
\begin{align*}
    Ra = \frac{\bold{g}\beta\Delta T L^3}{\nu^2}, \qquad \Delta T = (T_{left} - T_{right}),
\end{align*}
where $L$ is the width of the cavity. We were interested in computing the average steady state Nusselt number  $\overline{Nu}$ based on the horizontal heat flux as follows 
\begin{align*}
    \overline{Nu} &= \int_{0}^{1} \frac{H L}{\alpha} dx, \qquad H = \frac{\langle q_x \rangle_y}{A \Delta T}, \\[1.5ex]
    \langle q_x \rangle_y & = \int_{0}^{1}q_x dy, \qquad
    q_x = u_x T -\alpha \frac{\partial T}{\partial x},
\end{align*}
where $A$ is the domain area, and $u_x$ is the velocity in the $x$ direction. The Nusselt number was calculated at Rayleigh numbers of $Ra = 10^4, 10^5, 10^6,$ and $10^7$ which enabled the flow to remain laminar. At each Rayleigh number, four different grids of size $N\times N$ were considered with $N =8,16,32,$ and $64$. The only exception was Rayleigh number $Ra = 10^7$ as the $8 \times 8$ grid could not be converged for this mesh. The mesh used for each simulation was biased towards the walls using the mapping proposed by \cite{lowe2011finite}.  
\begin{align*}
    x_{refined} &= \Big ( x - \frac{1}{2\pi}(1-a)\sin(2\pi x) \Big) \\[1.5ex]
    y_{refined} &= \Big ( y - \frac{1}{2\pi}(1-b)\sin(2\pi y) \Big) \\[1.5ex]
    a &= \left( \, \overline{Nu} \,\right)^{-1}, \qquad b = \left( \, \overline{Nu} \, \right)^{-1/3}.
\end{align*}
 Note: in order to generate our meshes, we used the average Nusselt numbers reported in~\cite{dallmann2015finite}. On each mesh, Taylor-Hood elements of degree $k = 2$ and $3$ were used.

At the lowest Rayleigh number, the flow was dominated by a large central vortex seen in figure~\ref{fig:HC_4}. As the Rayleigh number was increased, this vortex disappeared and thin boundary layers developed on the left and right walls as seen in figures~\ref{fig:HC_R6} and~\ref{fig:HC_R7}. This is the same behavior observed by~\cite{dallmann2015finite}.
\begin{figure}[h!]
  \centering
  \includegraphics[width=1.0\linewidth]{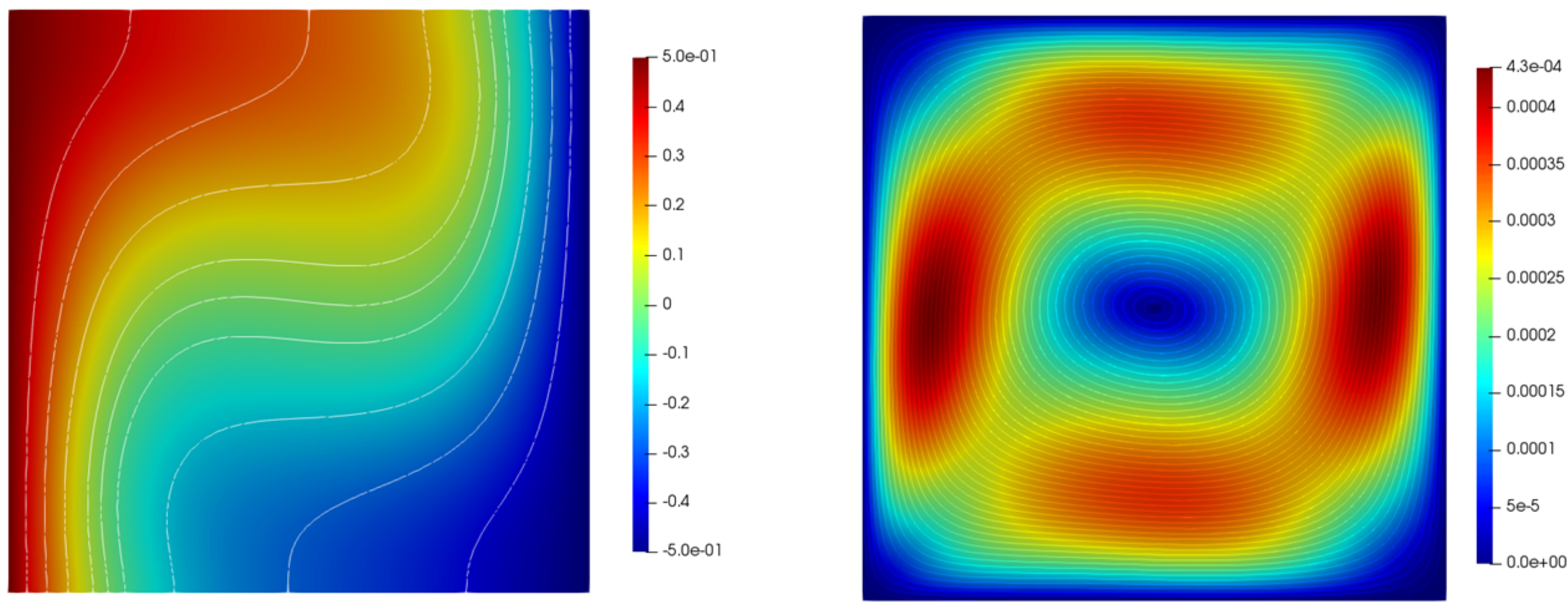}
  \vspace{-1.1cm}
\caption{Temperature (left) and velocity magnitude (right) for Rayleigh number $Ra = 10^4$. The $64\times64$ mesh with $k = 3$ was used to generate these results. }
\label{fig:HC_4}
\end{figure}

\begin{figure}[h!]
  \centering
  \includegraphics[width=1.0\linewidth]{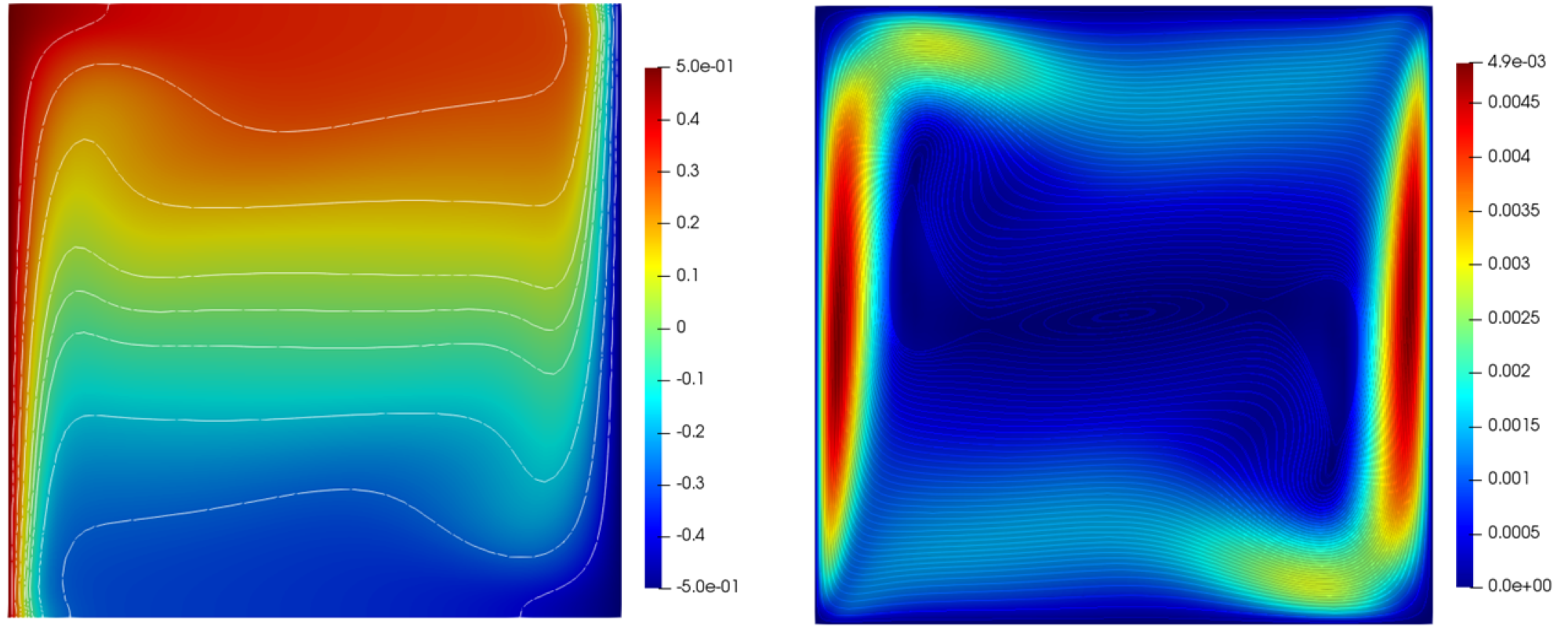}
  \vspace{-1.1cm}
\caption{Temperature (left) and velocity magnitude (right) for Rayleigh number $Ra = 10^6$. The $64\times64$ mesh with $k = 3$ was used to generate these results. }
\label{fig:HC_R6}
\end{figure}

\begin{figure}[h!]
  \centering
  \includegraphics[width=1.0\linewidth]{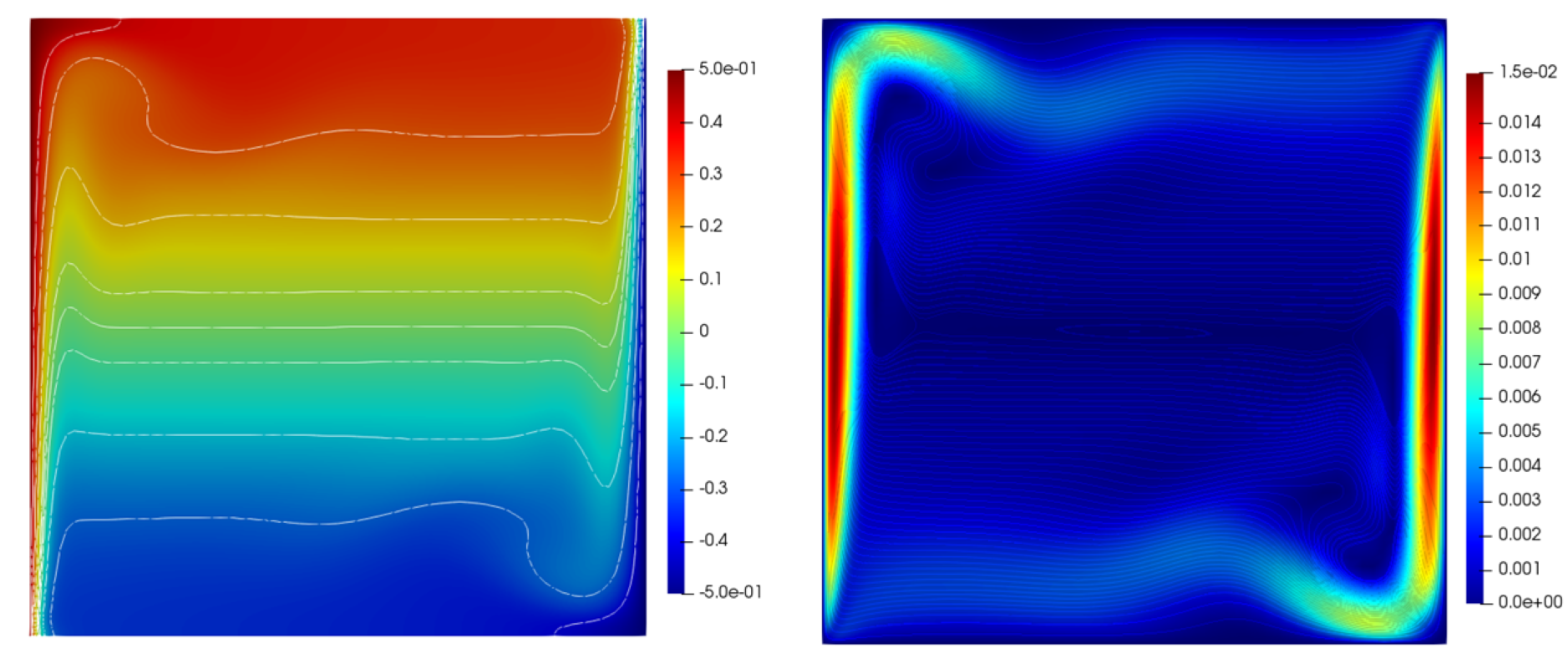}
  \vspace{-1.1cm}
\caption{Temperature (left) and velocity magnitude (right) for Rayleigh number $Ra = 10^7$. The $64\times64$ mesh with $k = 3$ was used to generate these results. }
\label{fig:HC_R7}
\end{figure}

\begin{table}[h]
\centering
\setlength{\extrarowheight}{2pt}
\begin{tabular}{|l|l|l|l|l|}
\hline
\multirow{2}{*}{$Ra$}                                          & \multirow{2}{*}{$N$} & \multirow{2}{*}{$\overline{Nu}$ ref.} & \multicolumn{2}{c|}{$\overline{Nu}$} \\[1.5ex]\cline{4-5}
                                                             &                    &                          & $k = 2$      & $k = 3$      \\ \hline
\multicolumn{1}{|c|}{\multirow{5}{*}{$10^4$}} & 8                  & -                        & 2.24480    & 2.24481    \\ \cline{2-5} 
\multicolumn{1}{|c|}{}                                       & 16                 & 2.24478                  & 2.24481    & 2.24481    \\ \cline{2-5} 
\multicolumn{1}{|c|}{}                                       & 32                 & 2.24481                  & 2.24481    & 2.24481    \\ \cline{2-5} 
\multicolumn{1}{|c|}{}                                       & 64                 & 2.24482                  & 2.24481    & 2.24481    \\ \hline 
\multirow{5}{*}{$10^5$}                       & 8                  & -                        & 4.52206    & 4.52162    \\ \cline{2-5} 
                                                             & 16                 & 4.52124                  & 4.52163    & 4.52163    \\ \cline{2-5} 
                                                             & 32                 & 4.52162                  & 4.52163    & 4.52163    \\ \cline{2-5} 
                                                             & 64                 & 4.52163                  & 4.52163    & 4.52163    \\ \hline
\multirow{5}{*}{$10^6$}                       & 8                  & -                          & 8.81679    & 8.82497    \\ \cline{2-5} 
                                                             & 16                 & 8.81573                  & 8.82514    & 8.82519    \\ \cline{2-5} 
                                                             & 32                 & 8.82502                  & 8.82520    & 8.82520    \\ \cline{2-5} 
                                                             & 64                 & 8.82519                  & 8.82520    & 8.82520    \\ \hline 
\multirow{4}{*}{$10^7$}                       & 16                 & 15.3718                  & 16.5190    & 16.5227    \\ \cline{2-5} 
                                                             & 32                 & 16.5156                  & 16.5229    & 16.5230    \\ \cline{2-5} 
                                                             & 64                 & 16.5230                  & 16.5230    & 16.5230    \\ \hline
\end{tabular}
\caption{Average Nusselt numbers compared at various Rayleigh numbers, mesh resolutions, and polynomial degrees. Reference values are taken from~\cite{dallmann2015finite}. }
\label{Dallman_table}
\end{table}
We also saw very close agreement with the average Nusselt number for all mesh resolutions and polynomial degrees as seen in table~\ref{Dallman_table}. This leads us to  believe  that our method is able to accurately capture purely buoyancy-driven flows. 

\subsection{Heated Cavity with Moving Wall Test}
The final test case was a heated cavity with one moving wall, i.e.~a mixed convection case. Here, the top wall moved at constant velocity $V_{lid}$ and was heated, while the bottom wall was cooled as proposed by~\cite{iwatsu1993mixed}. This case was run at a fixed Grashof number $Gr = 10^4$ along with varied Richardson numbers $Ri$, where 
\begin{align*}
    Gr = \frac{\bold{g} \beta L \Delta T}{\nu^2}, \qquad \Delta T = (T_{top} - T_{bottom}), \qquad
    Ri = \frac{Gr}{Re^2}, \qquad
    Re = \frac{V_{lid}L}{\nu}.
\end{align*}
We used the same fluid properties prescribed in the previous heated cavity case. In this case, we considered Richardson numbers $Ri = 0.01, 0.06$ and $1.0$. The domain was a box $\Omega = \left[0,1\right]^2$ with a uniform mesh that consisted of $64\times64$, $k = 2$ Taylor-Hood elements. The heated top wall was held at a constant temperature $T_{top} = 1$ K, while the bottom cold wall was held at $T_{bottom} = 0$ K, with the remaining walls having adiabatic boundary conditions. Gravity was again present in this case with $\bold{g}=(0,-1)^T ~\frac{m}{s^2}$. The quantity of interest for this case was again the average steady state Nusselt number $\overline{Nu}$, however for this case we were only interested in the Nusselt number along the top heated wall, referred to henceforth as $\overline{Nu}_{wall}$. We define the vertical heat flux $q_y$ and $\overline{Nu}_{wall}$ as
\begin{align*}
    \overline{Nu}_{wall} &= \frac{HL}{\alpha}, \qquad
    H = \frac{\langle q_y \rangle_x}{A \Delta T}, \\[1.5ex]
    q_y &= - \frac{\partial T}{\partial y}, \qquad 
    \langle q_y \rangle_x = \int_{0}^{1} [q_y]_{y=1} dx. 
\end{align*}
In this case, the best agreement with the reference data occurs for the largest Richardson number $Ri = 1$ as seen in table~\ref{MC_table}. 
\begin{table}[h!]
\centering
\setlength{\extrarowheight}{2pt}
\begin{tabular}{|l|l|l|}
\hline
$Ri$   & $\overline{Nu}_{wall}$ ref.  & $\overline{Nu}_{wall}$   \\ \hline
1.0  & 1.34    & 1.39 \\ \hline
0.06 & 3.62    & 3.87 \\ \hline
0.01 & 6.29    & 6.52 \\ \hline
\end{tabular}
\caption{Comparison of $\overline{Nu}_{wall}$ for our method with reference data at various Richardson numbers. We use~\cite{iwatsu1993mixed} for the reference values.}
\label{MC_table}
\end{table}
The other two cases with Richardson numbers $Ri = 0.06$ and $0.01$ show some deviation from the reference. We expected to see this deviation because in our method, unlike the reference's method, we included the viscous dissipation term inside the temperature equation. The viscous dissipation term became relevant at these Richardson numbers because as the Richardson numbers decreases the lid velocity increases, as seen in figures~\ref{fig:MC_Ri1},~\ref{fig:MC_Ri006}, and~\ref{fig:MC_Ri001}. The flow at these Richardson numbers behaves more like a lid-driven cavity flow as opposed to a natural convection dominated flow. Therefore at these higher lid speeds, the velocity gradients in the flow became non-negligible and have a more pronounced effect on the temperature field (via the dissipation term). This deviation in Nusselt number was absent from the heated cavity case because the flow in that case was subject to natural convection only, and thus possessed smaller velocity gradients as compared to the mixed convection case. 

\begin{figure}[h!]
  \centering
  \includegraphics[width=1.0\linewidth]{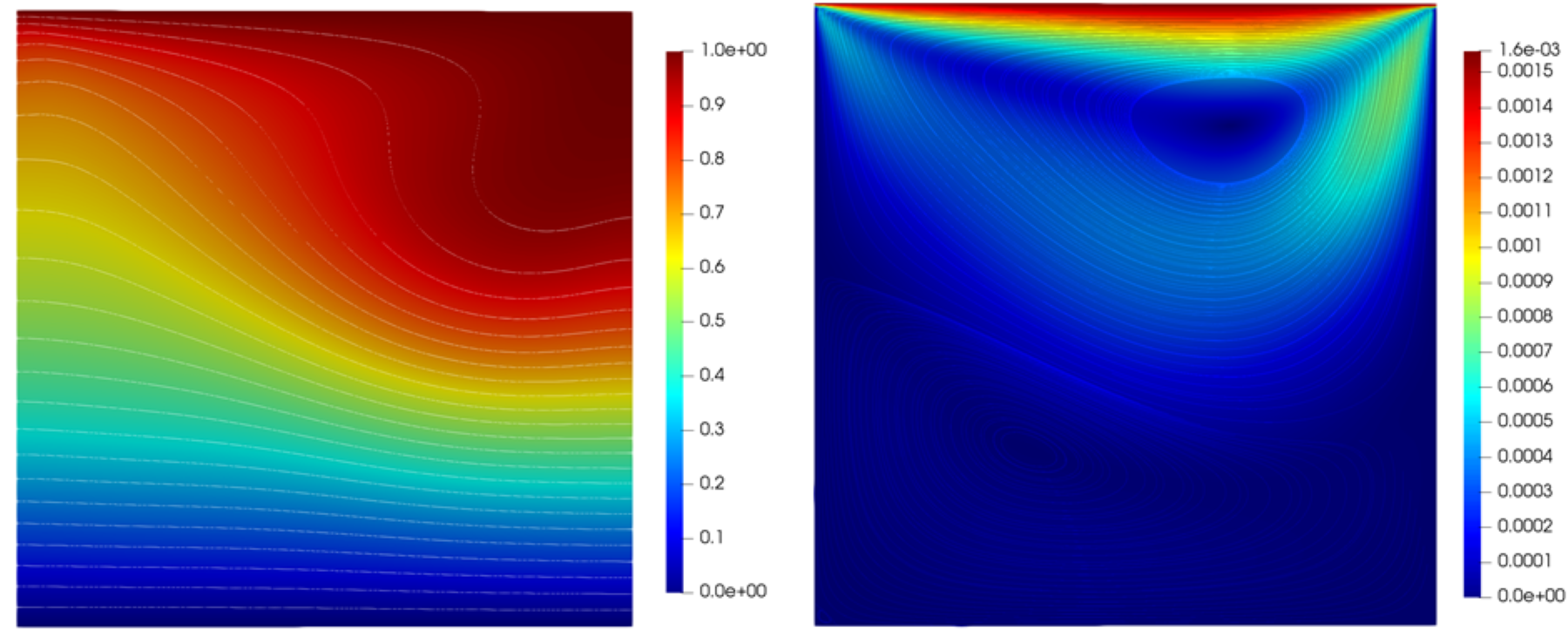}
  \vspace{-1.1cm}
\caption{Temperature contours (left) and velocity magnitude contours with streamlines (right)  for Richardson number $Ri = 1$. A $64 \times 64$ mesh with $k$ = 2 was used to generate these results.}
\label{fig:MC_Ri1}
\end{figure}

\begin{figure}[h!]
  \centering
  \includegraphics[width=1.0\linewidth]{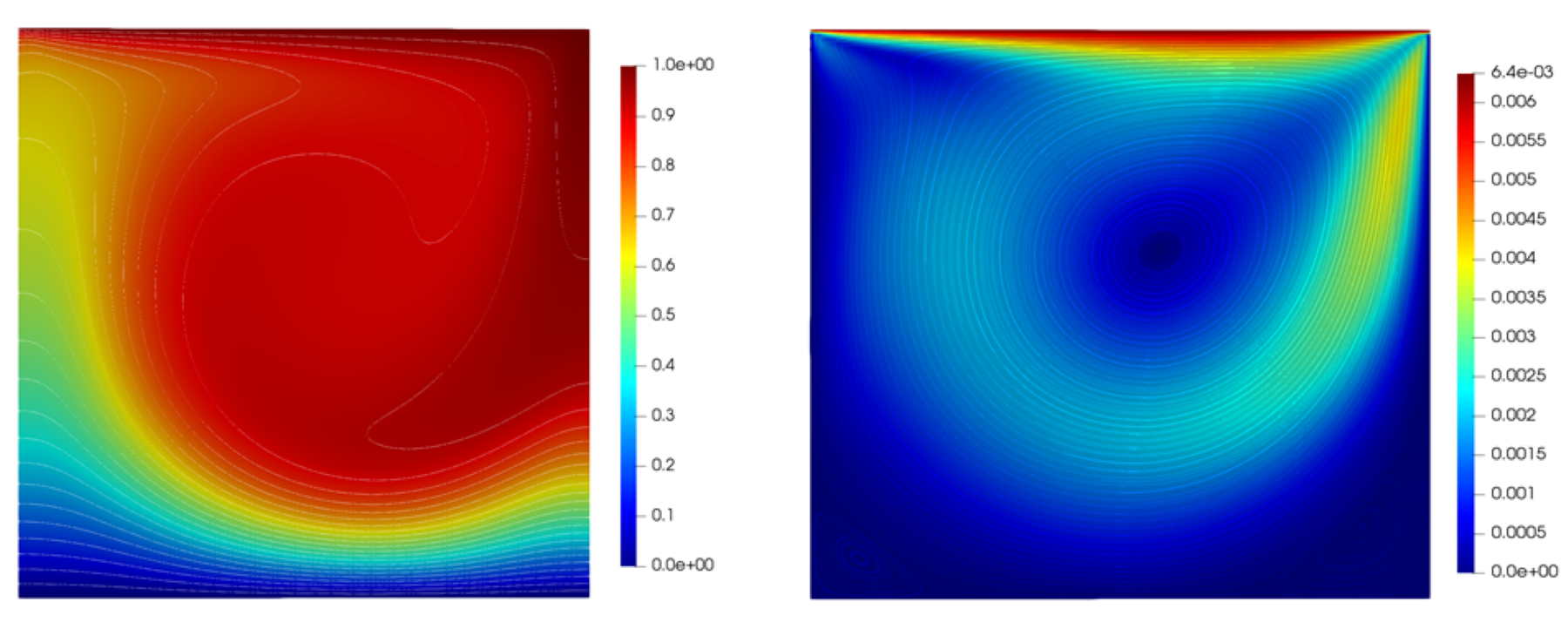}
  \vspace{-1.1cm}
\caption{Temperature contours (left) and velocity magnitude contours with streamlines (right)  for Richardson number $Ri = 0.06$. A $64 \times 64$ mesh with $k$ = 2 was used to generate these results.}
\label{fig:MC_Ri006}
\end{figure}

\begin{figure}[h!]
  \centering
  \includegraphics[width=1.0\linewidth]{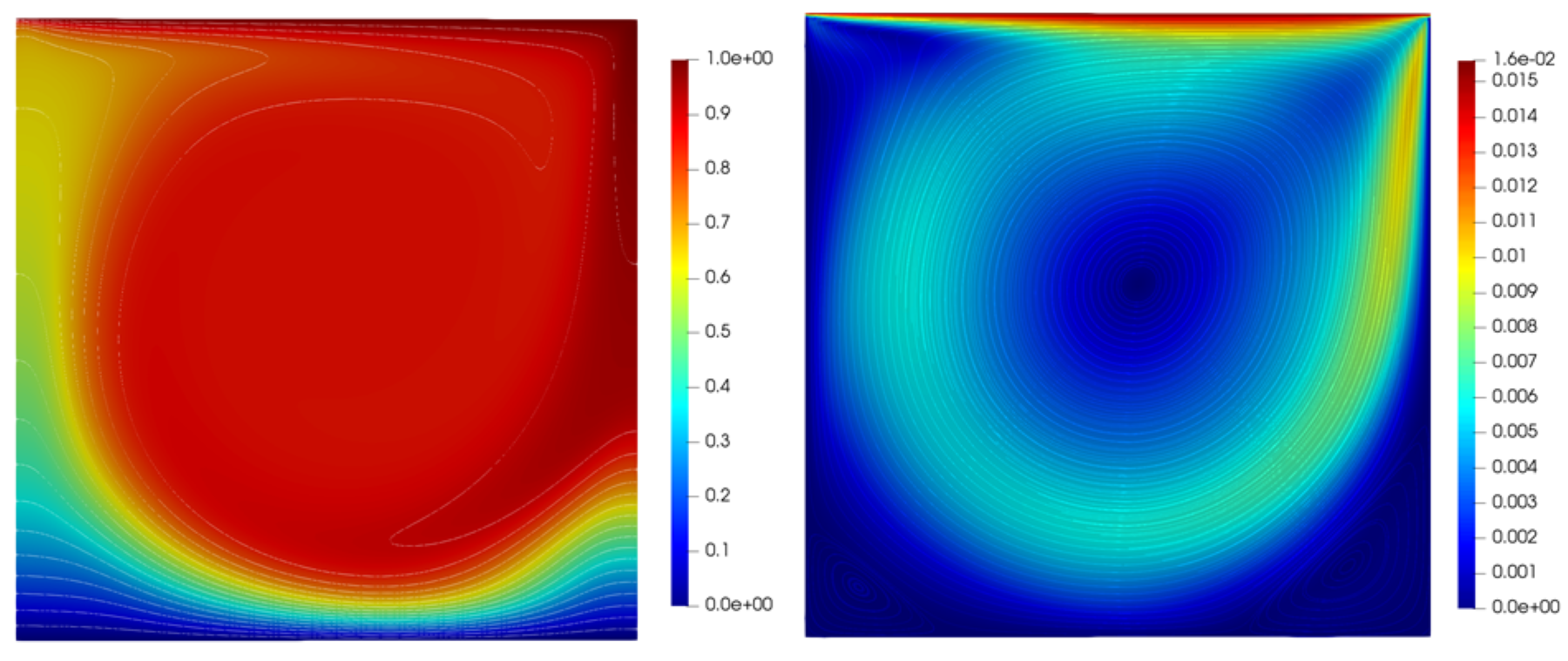}
  \vspace{-1.1cm}
\caption{Temperature contours (left) and velocity magnitude contours with streamlines (right) for Richardson number $Ri = 0.01$. A $64 \times 64$ mesh with $k$ = 2 was used to generate these results.}
\label{fig:MC_Ri001}
\end{figure}

\pagebreak

\section{Conclusion}

In the present study, the mixed methods first put forward by Chen and Williams~\cite{chen2020versatile} are extended to non-isothermal incompressible flows. The primary advantages that these new methods possess are their generality and flexibility, as they utilize the full compressible formulation of the stress tensor and the expanded formulation of the temperature equation (which retains the dilatational pressure work and viscous dissipation terms). In this paper, the new versatile methods are constructed for weakly  divergence-free Taylor-Hood elements, and pointwise divergence-free BDM and RT elements. Next, we rigorously derive a new condition that governs the L2-stability of the discrete temperature fields for these methods. Finally, the accuracy of the Taylor-Hood method is tested using three well-known cases from the literature; these tests are used to confirm the formal order of accuracy of the method, and demonstrate its performance on problems with natural and mixed convection. The analysis and numerical experiments in this article serve as a stepping stone towards the application of these methods to weakly-compressible and fully-compressible flows.


\section*{Declaration of Interests}

None.

\section*{Funding Sources}

This research did not receive any specific grant from funding agencies in the public, commercial, or not-for-profit sectors.   

\setcounter{section}{0}
\setcounter{theorem}{0}
\renewcommand{\thetheorem}{\thesection.\arabic{theorem}}
\setcounter{equation}{0}
\renewcommand{\theequation}{\thesection.\arabic{equation}}
\setcounter{figure}{0}
\setcounter{table}{0}

\appendix
\section*{Appendix}
\renewcommand{\thesection}{A} 
\setcounter{equation}{0}
\renewcommand{\theequation}{\thesection.\arabic{equation}}

\subsection{Numerical Fluxes} \label{numerical_flux_section}

We suggest the following formulations for the numerical fluxes
\begin{align}
\widehat{\bm{\sigma}}_{\text{inv}} &= \llcurve \ubold_h \rrcurve \otimes \llcurve \ubold_h \rrcurve + \llcurve \widetilde{p}_h \rrcurve \mathbb{I} + \zeta \left| \ubold_h \cdot \nbold_F \right| \llbracket \ubold_h \otimes \nbold \rrbracket, \\[1.5ex]
\widehat{\bm{\sigma}}_{\text{vis}} & = \llcurve \nabla_h \ubold_h + \nabla_h \ubold_{h}^{T} - \frac{2}{3} \left(\nabla \cdot \ubold_h \right) \mathbb{I} \rrcurve -\frac{\eta}{h_F} \llbracket \ubold_h \otimes \nbold \rrbracket, \\[1.5ex]
\widehat{\bm{\phi}}_{\text{inv}} &= \llcurve T_h \rrcurve \ubold_h + \delta \left| \ubold_h \cdot \nbold_F \right| \llbracket T_h \, \nbold \rrbracket,  \\[1.5ex]
\widehat{\bm{\phi}}_{\text{vis}} & = \llcurve \nabla_h T_h \rrcurve - \frac{\varepsilon}{h_F} \llbracket T_h \, \nbold \rrbracket, \\[1.5ex]
\widehat{\bm{\varphi}}_{\text{vis}} &= \llcurve \ubold_h \rrcurve, \\[1.5ex]
\widehat{\lambda}_{\text{vis}} &= \llcurve T_h \rrcurve,
\end{align}
where $\zeta$, $\eta$, $\delta$, and $\varepsilon$ are parameters which control the amount of dissipation introduced by the fluxes.

\subsection{Pointwise Divergence-Free Methods} \label{div_free_section}

In this section, we construct a pointwise divergence-free class of mixed methods for solving  Eqs.~\eqref{mass_cons} -- \eqref{energy_cons}. These methods can be formally stated as follows: i) identify function spaces $Q_h = Q_h^{DC}$, $R_h = R_h^C$, and $\bm{W}_h = \bm{W}_h^{RT}$ or $\bm{W}_h = \bm{W}_h^{BDM}$, ii) choose test functions $\left(q_h, r_h, \wbold_h \right)$ that span $Q_h \times R_h \times \bm{W}_h$, and iii) find unknowns $\left(\widetilde{p}_h, T_h, \ubold_h \right)$ in $Q_h  \times R_h \times \bm{W}_h$ that satisfy 
%
%
\begin{align}
& \ipt{ \nabla \cdot \ubold_h}{q_h} = 0, \label{mass_cons_disc_div_free}
\end{align}
\begin{align}
\nonumber & \ipt{\partial_t \ubold_h }{\wbold_h} - \ipt{ \ubold_h \otimes \ubold_h}{\nabla_h \wbold_h} - \ipt{\widetilde{p}_h}{\nabla \cdot \wbold_h} + \ipbt{\widehat{\bm{\sigma}}_{\text{inv}} \, \nbold}{\wbold_h} \\[1.5ex]  
\nonumber & + \nu_h \bigg[ \ipt{\nabla_h \ubold_h + \nabla_h \ubold_{h}^{T}}{\nabla_h \wbold_h} - \ipbt{\widehat{\bm{\sigma}}_{\text{vis}} \, \nbold}{\wbold_h} \\[1.5ex]
\nonumber & + \ipbt{\widehat{\bm{\varphi}}_{\text{vis}} - \ubold_h}{\left( \nabla_h \wbold_h + \nabla_h \wbold_{h}^{T} \right) \nbold} \bigg]
\\[1.5ex]
&= - \ipt{\beta_h T_h \gbold}{\wbold_h} + \ipt{\widetilde{\bm{f}}_{\ubold}}{\wbold_h}, \label{moment_cons_disc_div_free}
\end{align}
\begin{align}
& \nonumber \ipt{\partial_t T_h}{r_h} - \ipt{T_h \ubold_h}{\nabla_h r_h} + \ipbt{\widehat{\bm{\phi}}_{\text{inv}}  \cdot \nbold}{r_h}  \\[1.5ex]
\nonumber & + \alpha_h \gamma_h \Bigg[ \ipt{\nabla_h T_h}{\nabla_h r_h}  -\ipbt{\widehat{\bm{\phi}}_{\text{vis}}  \cdot \nbold}{r_h} + \ipbt{\widehat{\lambda}_{\text{vis}} - T_h}{\nabla_h r_h \cdot \nbold} \Bigg] \\[1.5ex]
&  = \frac{\nu_h}{C_v} \ipt{\left(\nabla_h \ubold_h + \nabla_h \ubold_{h}^{T} \right):\nabla_h \ubold_h}{r_h} + \ipt{\widetilde{f}_{T}}{r_h}. \label{energy_cons_disc_div_free}
\end{align}
This set of equations can be rewritten compactly as follows
\begin{align}
& b_h \left(\ubold_h, q_h \right) = 0, \label{incomp_form_two_A} 
\\[1.5ex]
\nonumber & \ipt{\partial_t \, \ubold_h}{\wbold_h} + c_h \left(\ubold_h; \ubold_h, \wbold_h \right) +\nu_h a_h \left(\ubold_h, \wbold_h \right)  - b_h \left( \wbold_h, \widetilde{p}_h \right)
\\[1.5ex]
& = - \ipt{\Xi \left(T_h\right)}{\wbold_h} + \ipt{\widetilde{\bm{f}}_{\ubold}}{\wbold_h}, \label{incomp_form_two_B}
\\[1.5ex]
&\nonumber \ipt{\partial_t T_h}{r_h} + \underline{c}_h\left(\ubold_h; T_h, r_h\right) + \alpha_h \gamma_h \, \underline{a}_h \left(T_h, r_h \right) 
\\[1.5ex]
&= \ipt{\Phi}{r_h} +  \ipt{\widetilde{f}_{T}}{r_h}. \label{incomp_form_two_C}
\end{align}
The operators $b_h$, $\Xi$, and $\underline{a}_h$ were previously defined in Eqs.~\eqref{bilinear_press_div}, \eqref{grav_term}, and \eqref{diff_bilinear_temp}. The remaining operators $c_h$, $a_h$, $\underline{c}_h$, and $\Phi$ can be written as follows
\begin{align}
c_h \left(\bm{\xi}_h; \vbold_h, \wbold_h \right) &= \ipt{\bm{\xi}_h \cdot \nabla_h \vbold_h}{\wbold_h}  \\[1.5ex]
& \nonumber - \iipbf{ \left( \bm{\xi}_h \cdot \nbold_F \right) \llbracket \vbold_h \rrbracket}{\llcurve \wbold_h \rrcurve} + \zeta \iipbf{\left| \bm{\xi}_h \cdot \nbold_F \right|\llbracket \vbold_h \rrbracket}{\llbracket \wbold_h \rrbracket}, 
\end{align}
\begin{align}
\nonumber a_h \left(\vbold_h, \wbold_h \right) &=  \ipt{ \nabla_h \vbold_h + \nabla_h \vbold_{h}^{T}}{\nabla_h \wbold_h}   -\ipbf{\llbracket \vbold_h \rrbracket}{\llcurve  \nabla_h \wbold_h + \nabla_h \wbold_{h}^{T} \rrcurve \nbold_F} \\[1.5ex]
& -\ipbf{\llbracket \wbold_h \rrbracket}{\llcurve  \nabla_h \vbold_h + \nabla_h \vbold_{h}^{T} \rrcurve \nbold_F} + \ipbf{\frac{\eta}{h_F} \llbracket \vbold_h \rrbracket}{\llbracket \wbold_h \rrbracket},
\label{diff_bilinear_div_free}
\end{align}
\begin{align}
\underline{c}_h \left(\bm{\xi}_h; \theta_h, r_h \right) &= \ipt{\bm{\xi}_h \cdot \nabla_h \theta_h}{r_h} \label{trilinear_temp_div_free}\\[1.5ex]
& \nonumber - \iipbf{ \left( \bm{\xi}_h \cdot \nbold_F \right) \llbracket \theta_h \rrbracket}{\llcurve r_h \rrcurve} + \delta \iipbf{\left| \bm{\xi}_h \cdot \nbold_F \right|\llbracket \theta_h \rrbracket}{\llbracket r_h \rrbracket}, 
\end{align}
\begin{align}
    \Phi \left(\vbold_h\right) &= \frac{\nu_h}{C_v}  \left(\left(\nabla_h \vbold_h + \nabla_h \vbold_{h}^{T} \right):\nabla_h \vbold_h\right).
\end{align}

\subsection{Supporting Results}

\begin{lemma}[Broken Norm Inequalities] \label{grad_inequality_lemma}
Suppose that $\wbold \in \bm{W}^{1,p} \left(\mathcal{T}_h\right)$ and $p \geq 2$. Then, the following inequalities hold
\begin{align}
    & \left\| \nabla_h \cdot \wbold \right\|_{L^{p}\left(\Omega\right)} \leq d^{(p-1)/p} \left\| \wbold \right\|_{\text{grad},p}, \label{ineq_one}  
    \\[1.5ex]
    &\left\| \nabla_h \wbold : \nabla_h \wbold \right\|_{L^{p/2}\left(\Omega\right)} \leq d^{2(p-2)/p} \left\| \wbold \right\|_{\text{grad},p}^2, \label{ineq_two}
    \\[1.5ex]
    &\left\| \nabla_h \wbold^T : \nabla_h \wbold \right\|_{L^{p/2}\left(\Omega\right)} \leq d^{2(p-2)/p} \left\| \wbold \right\|_{\text{grad},p}^2. \label{ineq_three}
\end{align}
\end{lemma}

\begin{proof}
Let us begin by noting that
\begin{align*}
    &\left\| \nabla_h \cdot \wbold \right\|_{L^{p}\left(\Omega\right)} \\[1.5ex]
    &= \left( \sum_{K \in \mathcal{T}_h} \int_{K} \left(\sum_{i}^{d} \left(\partial_i w_i \right) \right)^{p}  dV \right)^{1/p}  \leq \left( \sum_{K \in \mathcal{T}_h} \int_{K} \left(\sum_{i}^{d} \left|\partial_i w_i \right| \right)^{p}  dV \right)^{1/p} \\[1.5ex]
    &\leq d^{(p-1)/p} \left( \sum_{K \in \mathcal{T}_h} \int_{K} \sum_{i}^{d} \left(\partial_i w_i \right)^{p}  dV \right)^{1/p} \leq d^{(p-1)/p} \left( \sum_{K \in \mathcal{T}_h} \int_{K} \sum_{i,j}^{d} \left(\partial_j w_i \right)^{p}  dV \right)^{1/p}.
\end{align*}
Here, we have used the power mean inequality. Upon combining this result with the definition for the norm $\left\| \cdot \right\|_{\text{grad},p}$, we obtain the first result (Eq.~\eqref{ineq_one}).

The proofs of the remaining results (Eqs.~\eqref{ineq_two} and \eqref{ineq_three}) are virtually identical. In what follows, we will simply show the proof for Eq.~\eqref{ineq_two}. We begin by expanding the L$p/2$-norm
\begin{align}
    \left\| \nabla_h \wbold : \nabla_h \wbold \right\|_{L^{p/2}\left(\Omega\right)} = \left\| \sum_{i,j}^d \left( \partial_{j} w_i \right)^2 \right\|_{L^{p/2}\left(\Omega\right)} = \left( \sum_{K \in \mathcal{T}_h} \int_{K} \left( \sum_{i,j}^d \left( \partial_{j} w_i \right)^2 \right)^{p/2} dV \right)^{2/p}. \label{norm_inequality_one}
\end{align}
Then, by the power mean inequality, we have
\begin{align}
    \nonumber \sum_{K \in \mathcal{T}_h} \int_{K} \left( \sum_{i,j}^d \left( \partial_{j} w_i \right)^2 \right)^{p/2} dV &\leq d^{(p-2)} \left( \sum_{K \in \mathcal{T}_h} \int_{K}  \sum_{i,j}^d \left| \partial_{j} w_i \right|^p dV \right)   
    \\[1.5ex]
    \left( \sum_{K \in \mathcal{T}_h} \int_{K} \left( \sum_{i,j}^d \left( \partial_{j} w_i \right)^2 \right)^{p/2} dV \right)^{2/p} & \leq d^{2(p-2)/p} \left( \sum_{K \in \mathcal{T}_h} \int_{K}  \sum_{i,j}^d \left| \partial_{j} w_i \right|^p dV \right)^{2/p}. \label{ineq_result_one}
\end{align}
Upon combining this result with Eq.~\eqref{norm_inequality_one}, and the definition for the norm $\left\| \cdot \right\|_{\text{grad},p}$, we obtain the inequality in Eq.~\eqref{ineq_two}.

\end{proof}


{\footnotesize\bibliography{technical-refs}}

\begin{thebibliography}{10}
\expandafter\ifx\csname url\endcsname\relax
  \def\url#1{\texttt{#1}}\fi
\expandafter\ifx\csname urlprefix\endcsname\relax\def\urlprefix{URL }\fi
\expandafter\ifx\csname href\endcsname\relax
  \def\href#1#2{#2} \def\path#1{#1}\fi

\bibitem{oberbeck1879warmeleitung}
A.~Oberbeck, {\"U}ber die w{\"a}rmeleitung der fl{\"u}ssigkeiten bei
  ber{\"u}cksichtigung der str{\"o}mungen infolge von temperaturdifferenzen,
  Annalen der Physik 243~(6) (1879) 271--292.

\bibitem{boussinesq1903theorie}
J.~Boussinesq, Th{\'e}orie analytique de la chaleur mise en harmonic avec la
  thermodynamique et avec la th{\'e}orie m{\'e}canique de la lumi{\`e}re,
  Vol.~2, Gauthier-Villars, 1903.

\bibitem{zeytounian2003joseph}
R.~K. Zeytounian, Joseph {Boussinesq} and his approximation: a contemporary
  view, Comptes Rendus Mecanique 331~(8) (2003) 575--586.

\bibitem{dallmann2016stabilized}
H.~Dallmann, D.~Arndt, Stabilized finite element methods for the
  {Oberbeck}--{Boussinesq} model, Journal of Scientific Computing 69~(1) (2016)
  244--273.

\bibitem{laskaris1975finite}
T.~E. Laskaris, Finite-element analysis of compressible and incompressible
  viscous flow and heat transfer problems, The physics of fluids 18~(12) (1975)
  1639--1648.

\bibitem{young1976steady}
D.-L. Young, R.~H. Gallagher, J.~A. Liggett, Steady stratified circulation in a
  cavity, Journal of the Engineering Mechanics Division 102~(1) (1976) 1--17.

\bibitem{young1976unsteady}
D.-L. Young, R.~H. Gallagher, J.~A. Liggett, Steady stratified circulation in a
  cavity, Journal of the Engineering Mechanics Division 102~(1) (1976)
  1009--1023.

\bibitem{tabarrok1977finite}
B.~Tabarrok, R.~C. Lin, Finite element analysis of free convection flows,
  International Journal of Heat and Mass Transfer 20~(9) (1977) 945--952.

\bibitem{gartling1977convective}
D.~K. Gartling, Convective heat transfer analysis by the finite element method,
  Computer Methods in Applied Mechanics and Engineering 12~(3) (1977) 365--382.

\bibitem{marshall1978natural}
R.~S. Marshall, J.~C. Heinrich, O.~Zienkiewicz, Natural convection in a square
  enclosure by a finite-element, penalty function method using primitive fluid
  variables, Numerical Heat Transfer, Part B: Fundamentals 1~(3) (1978)
  315--330.

\bibitem{reddy1980comparison}
J.~Reddy, A.~Satake, A comparison of a penalty finite element model with the
  stream function-vorticity model of natural confection in enclosures, Journal
  of Heat Transfer 102 (1980) 859.

\bibitem{hughes1987new}
T.~J. Hughes, L.~P. Franca, A new finite element formulation for computational
  fluid dynamics: {VII}. {The Stokes} problem with various well-posed boundary
  conditions: symmetric formulations that converge for all velocity/pressure
  spaces, Computer Methods in Applied Mechanics and Engineering 65~(1) (1987)
  85--96.

\bibitem{hughes1989new}
T.~J. Hughes, L.~P. Franca, G.~M. Hulbert, A new finite element formulation for
  computational fluid dynamics: Viii. the galerkin/least-squares method for
  advective-diffusive equations, Computer Methods in Applied Mechanics and
  Engineering 73~(2) (1989) 173--189.

\bibitem{baiocchi1993virtual}
C.~Baiocchi, F.~Brezzi, L.~P. Franca, Virtual bubbles and
  {Galerkin}-least-squares type methods, Computer Methods in Applied Mechanics
  and Engineering 105~(1) (1993) 125--141.

\bibitem{tang1994least}
L.~Tang, T.~Tsang, A least-squares finite element method for doubly-diffusive
  convection, International Journal of Computational Fluid Dynamics 3~(1)
  (1994) 1--17.

\bibitem{tang1997temporal}
L.~Q. Tang, T.~T. Tsang, Temporal, spatial and thermal features of {3-D}
  {Rayleigh}-{B{\'e}nard} convection by a least-squares finite element method,
  Computer Methods in Applied Mechanics and Engineering 140~(3-4) (1997)
  201--219.

\bibitem{reddy2010finite}
J.~N. Reddy, D.~K. Gartling, The Finite Element Method in Heat Transfer and
  Fluid Dynamics, CRC press, 2010.

\bibitem{dallmann2015finite}
H.~Dallmann, Finite element methods with local projection stabilization for
  thermally coupled incompressible flow, Ph.D. thesis, Nieders{\"a}chsische
  Staats-und Universit{\"a}tsbibliothek G{\"o}ttingen (2015).

\bibitem{boland1990error}
J.~Boland, W.~Layton, Error analysis for finite element methods for steady
  natural convection problems, Numerical Functional Analysis and Optimization
  11~(5-6) (1990) 449--483.

\bibitem{boland1990analysis}
J.~Boland, W.~Layton, An analysis of the finite element method for natural
  convection problems, Numerical Methods for Partial Differential Equations
  6~(2) (1990) 115--126.

\bibitem{dorok1994aspects}
O.~Dorok, W.~Grambow, L.~Tobiska, Aspects of finite element discretizations for
  solving the {Boussinesq} approximation of the {Navier}-{Stokes} equations,
  in: Numerical Methods for the Navier-Stokes Equations, Springer, 1994, pp.
  50--61.

\bibitem{bernardi1995couplage}
C.~Bernardi, B.~M{\'e}tivet, B.~Pernaud-Thomas, Couplage des {\'e}quations de
  {Navier}-{Stokes} et de la chaleur: le modele et son approximation par
  {\'e}l{\'e}ments finis, ESAIM: Mathematical Modelling and Numerical
  Analysis-Mod{\'e}lisation Math{\'e}matique et Analyse Num{\'e}rique 29~(7)
  (1995) 871--921.

\bibitem{codina2010finite}
R.~Codina, J.~Principe, M.~{\'A}vila, Finite element approximation of turbulent
  thermally coupled incompressible flows with numerical sub-grid scale
  modelling, International Journal of Numerical Methods for Heat \& Fluid Flow
  20~(5) (2010) 492--516.

\bibitem{lowe2012projection}
J.~L\"{o}we, G.~Lube, A projection-based variational multiscale method for
  large-eddy simulation with application to non-isothermal free convection
  problems, Mathematical Models and Methods in Applied Sciences 22~(02) (2012)
  1150011.

\bibitem{roos2008robust}
H.-G. Roos, M.~Stynes, L.~Tobiska, Robust numerical methods for singularly
  perturbed differential equations: convection-diffusion-reaction and flow
  problems, Vol.~24, Springer Science \& Business Media, 2008.

\bibitem{matthies2015local}
G.~Matthies, L.~Tobiska, Local projection type stabilization applied to
  inf--sup stable discretizations of the {Oseen} problem, IMA Journal of
  Numerical Analysis 35~(1) (2015) 239--269.

\bibitem{brooks1982streamline}
A.~N. Brooks, T.~J. Hughes, Streamline upwind/{Petrov}-{Galerkin} formulations
  for convection dominated flows with particular emphasis on the incompressible
  {Navier}-{Stokes} equations, Computer Methods in Applied Mechanics and
  Engineering 32~(1-3) (1982) 199--259.

\bibitem{hughes1986new}
T.~J. Hughes, M.~Mallet, M.~Akira, A new finite element formulation for
  computational fluid dynamics: {II}. {Beyond} {SUPG}, Computer Methods in
  Applied Mechanics and Engineering 54~(3) (1986) 341--355.

\bibitem{hughes1987recent}
T.~J. Hughes, Recent progress in the development and understanding of {SUPG}
  methods with special reference to the compressible {Euler} and
  {Navier}-{Stokes} equations, International Journal for Numerical Methods in
  Fluids 7~(11) (1987) 1261--1275.

\bibitem{franca1988two}
L.~P. Franca, T.~J. Hughes, Two classes of mixed finite element methods,
  Computer Methods in Applied Mechanics and Engineering 69~(1) (1988) 89--129.

\bibitem{rebollo2018high}
T.~C. Rebollo, M.~G. M{\'a}rmol, F.~Hecht, S.~Rubino, I.~S. Mu{\~n}oz, A
  high-order local projection stabilization method for natural convection
  problems, Journal of Scientific Computing 74~(2) (2018) 667--692.

\bibitem{de2019grad}
J.~de~Frutos, B.~Garc{\'\i}a-Archilla, J.~Novo, Grad-div stabilization for the
  time-dependent {Boussinesq} equations with inf-sup stable finite elements,
  Applied Mathematics and Computation 349 (2019) 281--291.

\bibitem{oyarzua2017analysis}
R.~Oyarz{\'u}a, P.~Z{\'u}{\~n}iga, Analysis of a conforming finite element
  method for the {Boussinesq} problem with temperature-dependent parameters,
  Journal of Computational and Applied Mathematics 323 (2017) 71--94.

\bibitem{almonacid2018mixed}
J.~A. Almonacid, G.~N. Gatica, R.~Oyarz{\'u}a, A mixed--primal finite element
  method for the boussinesq problem with temperature-dependent viscosity,
  Calcolo 55~(3) (2018) 36.

\bibitem{allendes2018divergence}
A.~Allendes, G.~R. Barrenechea, C.~Naranjo, A divergence-free low-order
  stabilized finite element method for a generalized steady state {Boussinesq}
  problem, Computer Methods in Applied Mechanics and Engineering 340 (2018)
  90--120.

\bibitem{almonacid2020new}
J.~A. Almonacid, G.~N. Gatica, R.~Oyarz{\'u}a, R.~Ruiz-Baier, A new mixed
  finite element method for the n-dimensional {Boussinesq} problem with
  temperature-dependent viscosity, Networks \& Heterogeneous Media 15~(2)
  (2020) 215.

\bibitem{chen2020versatile}
X.~Chen, D.~M. Williams, Versatile mixed methods for the incompressible
  {Navier-Stokes} equations, Computers and Mathematics with Applications.

\bibitem{boffi2013mixed}
D.~Boffi, F.~Brezzi, M.~Fortin, Mixed Finite Element Methods and Applications,
  Vol.~44, Springer, 2013.

\bibitem{DiPietro11}
D.~A. Di~Pietro, A.~Ern, Mathematical Aspects of Discontinuous Galerkin
  Methods, Vol.~69, Springer Science \& Business Media, Berlin Heidelberg,
  2011.

\bibitem{arndt2015local}
D.~Arndt, H.~Dallmann, G.~Lube, Local projection {FEM} stabilization for the
  time-dependent incompressible {Navier}-{Stokes} problem, Numerical Methods
  for Partial Differential Equations 31~(4) (2015) 1224--1250.

\bibitem{alnaes2015fenics}
M.~S. Alnaes, J.~Blechta, J.~Hake, A.~Johansson, B.~Kehlet, A.~Logg,
  C.~Richardson, J.~Ring, M.~E. Rognes, G.~N. Wells, The {FEniCS} project
  version 1.5, Archive of Numerical Software 3~(100) (2015) 9--23.

\bibitem{lowe2011finite}
J.~L{\"o}we, Eine finite-elemente-methode f{\"u}r nicht-isotherme
  inkompressible str{\"o}mungsprobleme, Ph.D. thesis, Nieders{\"a}chsische
  Staats-und Universit{\"a}tsbibliothek G{\"o}ttingen (2011).

\bibitem{iwatsu1993mixed}
R.~Iwatsu, J.~M. Hyun, K.~Kuwahara, Mixed convection in a driven cavity with a
  stable vertical temperature gradient, International Journal of Heat and Mass
  Transfer 36~(6) (1993) 1601--1608.

\end{thebibliography}

\end{document}